\documentclass{article}

\usepackage{arxiv}

\usepackage[utf8]{inputenc} 
\usepackage[T1]{fontenc}    
\usepackage{hyperref}       
\usepackage{url}            
\usepackage{booktabs}       
\usepackage{amsfonts}       
\usepackage{nicefrac}       
\usepackage{microtype}      
\usepackage{lipsum}
\usepackage{graphicx}
\graphicspath{ {./images/} }

\usepackage{amsmath,amsfonts,amssymb,amsthm}
\usepackage{caption}
\usepackage{subcaption}
\usepackage[]{graphicx}
\usepackage{thm-restate}
\usepackage[ruled,linesnumbered,noend]{algorithm2e}
\usepackage{esvect}

\newtheorem{theorem}{Theorem}
\newtheorem{lemma}{Lemma}
\newtheorem{definition}{Definition}

\usepackage{xspace}

\newcommand{\panorama}{\textsf{PanORAMa}\xspace}
\newcommand{\optorama}{\textsf{OptORAMa}\xspace}
\newcommand{\rankoram}{\textsf{RankORAM}\xspace}
\newcommand{\pathoram}{\textsf{PathORAM}\xspace}
\newcommand{\sroram}{\textsf{SR-ORAM}\xspace}
\newcommand{\bucket}{\textsf{BucketORAM}\xspace}
\newcommand{\tworam}{\textsf{TWORAM}\xspace}
\newcommand{\ringoram}{\textsf{RingORAM}\xspace}

\newcommand{\hierarchical}{\textsf{HierarchicalORAM}\xspace}
\newcommand{\partition}{\textsf{PartitionORAM}\xspace}
\newcommand{\readop}{\textsf{read}}
\newcommand{\writeop}{\textsf{write}}

\newcommand{\access}{\texttt{access}\xspace}
\newcommand{\shuffle}{\texttt{shuffle}\xspace}
\newcommand{\evict}{\texttt{evict}\xspace}
\newcommand{\level}{\texttt{level}\xspace}
\newcommand{\rebuild}{\texttt{rebuild}\xspace}
\newcommand{\shortqueue}{\textsf{shortQueueShuffle}\xspace}
\newcommand{\cacheshuff}{\textsf{cacheShuffle}\xspace}

\newcommand{\rank}{\texttt{rank}\xspace}
\newcommand{\membership}{\texttt{membership}\xspace}
\newcommand{\indll}{\texttt{index}\xspace}
\newcommand{\posll}{\texttt{position}\xspace}
\newcommand{\historical}{\textsf{historicalMembership}\xspace}
\newcommand{\counter}{\textsf{compressedCounters}\xspace}
\newcommand{\arraymap}{\textsf{array}\xspace}
\newcommand{\blockbuffer}{\textsf{blockBuffer}\xspace}

\usepackage{pifont}
\newcommand{\cmark}{\ding{51}}%

\newcommand{\etal}{\emph{et al.}\xspace}

\usepackage{xcolor}

\title{Single Round-trip Hierarchical ORAM \\ via Succinct Indices}

\author{
 William Holland \\
  CSIRO's Data61 \\
  Melbourne, Australia \\
  \texttt{rayne.holland@data61.csiro.au} \\
   \And
    Olga Ohrimenko \\
  School of Coumputing and Information Systems\\
  The University of Melbourne\\
  Melbourne, Australia \\
  \texttt{oohrimenko@unimelb.edu.au} \\
  \And
Anthony Wirth \\
  School of Coumputing and Information Systems\\
  The University of Melbourne\\
  Melbourne, Australia \\
  \texttt{awirth@unimelb.edu.au} \\
}

\begin{document}
\maketitle
\begin{abstract}
Access patterns to data stored remotely create a side channel that is known to leak information even if the content of the data is encrypted.
To protect against access pattern leakage, Oblivious RAM is a cryptographic primitive that obscures the (actual) access trace at the expense of additional access and periodic shuffling of the server's contents.
A class of ORAM solutions, known as Hierarchical ORAM, has achieved theoretically \emph{optimal} logarithmic bandwidth overhead.
However, to date, Hierarchical ORAMs are seen as only theoretical artifacts.
This is because they require a large number of communication round-trips to locate (shuffled) elements at the server and involve complex building blocks such as cuckoo hash tables.

To address the limitations of Hierarchical ORAM schemes in practice, we introduce Rank ORAM; the first Hierarchical ORAM that can retrieve data with a single round-trip of communication (as compared to a logarithmic number in previous work).
To support non-interactive communication, we introduce a \emph{compressed} client-side data structure that stores, implicitly, the location of each element at the server.
In addition, this location metadata enables a simple protocol design that dispenses with the need for complex cuckoo hash tables.

Rank ORAM requires asymptotically smaller memory than existing (non-Hierarchical) state-of-the-art practical ORAM schemes (e.g., Ring ORAM) while maintaining comparable bandwidth performance.
Our experiments on real network file-system traces demonstrate a reduction in client memory, against existing approaches, of a factor of~$100$.
For example, when  {outsourcing} a database of $17.5$TB,
required client-memory is only $290$MB vs. $40$GB for standard approaches.
\end{abstract}

\clearpage

\section{Introduction}
\label{sec:intro}

In remote storage settings, where a client outsources their data to a server, encryption is an important mechanism for protecting data content.
Unfortunately encryption alone cannot protect against all vulnerabilities.
For instance, the order in which a client accesses its outsourced data, known as the \textit{access pattern}, can leak sensitive information.
Access pattern vulnerabilities have been demonstrated in a number of domains.
These include, but are not limited to, leakage through page fault patterns in secure processors~\cite{shinde2016preventing,wang2017leaky,xu2015controlled},
through SQL query patterns on encrypted outsourced databases \cite{islam2012access,kellaris2016generic}
and through search patterns, resulting in query recovery attacks, in searchable encryptions \cite{cash2015leakage,islam2012access}.

To mitigate access pattern leakage, Goldreich and Ostrovsky introduced the notion of Oblivious RAM (ORAM) \cite{goldreich1996software}.
An ORAM scheme transforms a sequence of {original (referred to as virtual)} accesses into a sequence of physical accesses that is independent of the original sequence.
This transformation eliminates information leakage in the access trace.
Ultimately, an adversary must not be able to distinguish between the access patterns produced by an ORAM on an arbitrary pair of input sequences of the same length.

The transformation induced by an ORAM obfuscates the original access pattern by performing both additional dummy accesses and periodic shuffling of the server's contents.
Since these operations incur bandwidth overhead,
the principal aim of ORAM research is to minimize this bandwidth overhead while maintaining adequate privacy. 
However, it is known that one cannot achieve performance better than a logarithmic overhead (in terms of the number of outsourced blocks) due to the established lower bound~\cite{goldreich1996software,larsen2018yes}.

A class of solutions known as Hierarchical ORAMs \sloppy (\hierarchical) \cite{asharov2020optorama,asharov2021oblivious}
has achieved optimality in terms of bandwidth. 
The \hierarchical was first introduced by Goldreich and Ostrovsky \cite{goldreich1996software} and has enjoyed a long line of improvements \cite{chan2017oblivious,kushilevitz2012security,lu2013distributed,goodrich2011privacy,goodrich2012privacy,patel2018panorama,pinkas2010oblivious,williams2012single} and variations \cite{boneh2011remote,fletcher2015bucket,nagarajan2019rho} over the past two decades. 
However, these constructions have poor \textit{concrete}\footnote{Throughout this work ``concrete'' refers to the actual bandwidth cost, unobscured by complexity notation. } 
bandwidth performance and require multiple round-trips of communication per access.
As a result, they have not been adopted in practice.
Further, the processes of rebuilding and shuffling the server's contents are often complex and have been shown to be error prone \cite{hemenway2021alibi,kushilevitz2012security}.

In this paper, we address the limitations of \hierarchical schemes in practice.
We construct \rankoram, the first \hierarchical that performs a \emph{single} round-trip access without the help of server-side computation.
The main insight of our construction is that knowledge of \textit{when} a block was last accessed enables non-interactive queries to a Hierarchical ORAM.
However, storing this recency information for \textit{all} data blocks is expensive and may be prohibitive in memory-constrained environments. 
We overcome this challenge with a novel succinct data structure, called \historical, that keeps recency information in a compressed format at the client.
\historical
builds on work on approximate recency queries by Holland \textit{et al.}~\cite{holland2020recency}
and exploits the periodic rebuilds of Hierarchical ORAMs.

\rankoram not only supports low latency, but, additionally, permits simplified rebuild operations.
For example, it is standard practice for Hierarchical ORAMs to employ oblivious cuckoo hash tables as the core data structure.
However, constructing oblivious cuckoo hash tables is expensive in practice.
\historical enables \rankoram to replace the cuckoo hash tables with simpler permuted arrays.
Thus, the use of \historical, as an additional structure at the client, induces a trade-off between client memory and bandwidth efficiency.
Notably, larger client memory can be supported in cloud computing applications where the service can rely on the client memory, either RAM or disk, of a desktop computer \cite{dautrich2014burst}.

The following theorem captures the performance of \rankoram.

\begin{theorem}
\textup{\rankoram} is an oblivious RAM
that
stores a database of $n$ blocks of $B$ bits, requires $\mathcal{O}(n +\sqrt{n}\cdot B)$ bits of private memory, performs \textup{\access} in a single round-trip, \textup{\rebuild} in an amortized constant number of round-trips and observes an amortized bandwidth overhead of $4\log n$ blocks. 
\label{thm:main}
\end{theorem}

{ ORAM schemes can be categorized into two types: \textit{classical} and \textit{extended}.
Classical ORAM schemes, which includes the class of \hierarchical, assume a client memory allocation of $\mathcal{O}(B)$ bits (\textit{i.e.}, a constant number of blocks) and no server-side computation.
In contrast, an extended ORAM scheme relaxes one or both of these assumptions.
\rankoram provides several performance advantages with respect to both types of ORAM, as summarised in Tables~\ref{tab:hierarchical} and~\ref{tab:practical}.
Compared to classical ORAM schemes, \rankoram
is the only non-interactive protocol (see Table~\ref{tab:hierarchical}). }
Though it requires more memory on the client side, the allocation is small when compared to the size of the outsourced database in practice ({e.g.,} $290$MB for a $17.5$TB database).
In comparison to extended ORAM schemes, including popular tree-based constructions \cite{stefanov2013path}, \rankoram achieves a state-of-the-art bandwidth overhead with a client-memory allocation \textit{smaller} than comparable approaches (see Table~\ref{tab:practical}).
Notably, relative to Partition ORAM \cite{stefanov2011towards} (\partition) and Ring ORAM \cite{ren2014ring} (\ringoram), we reduce the memory allocation at the client by a factor of $\mathcal{O}(\log n)$ --- making it the first single-round trip ORAM to achieve that without server side computation.
In comparison to the recursive Path ORAM, an interactive and small memory predecessor of \ringoram, we reduce the bandwidth overhead by a factor of~4 and the number of round trips by~$\log n$.

In summary, our contributions are as follows.
\begin{itemize}
    \item 
    We propose \rankoram, the \textit{first} \hierarchical that achieves the combination of a single round-trip per access, a constant number of round-trips for rebuilds and a low concrete bandwidth overhead  (Table~\ref{tab:hierarchical}).
    \item \rankoram is supported by a novel client-side data structure, \historical, that compresses the location metadata of blocks at the server.
    \historical can be used to reduce the number of round-trips per  access for \textit{any} hierarchical ORAM solution.
    \item
    The client memory allocation of \rankoram\  is asymptotically smaller than prior ORAMs that achieve state-of-the-art latency and concrete bandwidth overhead~(Table~\ref{tab:practical}).
    {\rankoram is the \emph{first} ORAM to achieve, \emph{without} server side computation, a single round trip of communication {and} $\mathcal{O}(n)$ bits of client memory.}
    \item 
    Our experiments in Section~\ref{sec:experiments}, conducted on real network file system traces, demonstrate, for a~$4$KB block size, a reduction in memory for \rankoram by a factor of~$100$ against a non-compressed structure.
\end{itemize}



{We continue with related work (Section~\ref{sec:related}) and some preliminaries (Section~\ref{sec:prelim}).
A technical overview of our solution is presented in Section~\ref{sec:overview}.
It sketches the core ideas that drive our algorithms and highlights how \historical enables the performance gains achieved by \rankoram.
The low-level details of \rankoram and \historical are presented, respectively, in Section~\ref{sec:rank} and Section~\ref{SEC:CLIENT}.
Section~\ref{sec:security} provides an assessment of the performance and security of \rankoram.
Finally, Section~\ref{sec:experiments} contains an experimental evaluation of \rankoram against prior state-of-the-art solutions.}

\begin{table*}[t]
    \centering
    \caption{
    {Asymptotical comparison of \textit{classical} ORAM solutions with \rankoram. 
    ORAMs are parameterized as outsourcing $n$ blocks of $B$ bits each.}
    \rankoram is the \emph{only} solution that achieves a single round-trip per \access,
    while all other solutions in the table require $\Omega(\log n)$ round-trips.
    }
    \begin{tabular}{|c||c|c|c|c|} \hline
        & client storage 
            & bandwidth 
                & hash table 
                    & \hierarchical
                    \\ \hline \hline 
    Square Root \cite{goldreich1996software} 
        & $\mathcal{O}(B)$
            & $\mathcal{O}(\sqrt{n})$
                & Permuted array
                    & 
                    \\ \hline
    Hierarchical \cite{goldreich1996software}   
        & $\mathcal{O}(B)$
            &  $\mathcal{O}(\log^3 n)$
                & Balls-in-bins
                    & \cmark
                     \\ \hline
    Cuckoo \cite{goodrich2011privacy,chan2017oblivious}
        & $\mathcal{O}(B)$
            & $\mathcal{O}(\log^2 n/\log \log n)$
                & Cuckoo
                    & \cmark
                     \\ \hline
    Chan~\textit{et al.} \cite{chan2017oblivious}
        & $\mathcal{O}(B)$
            & $\mathcal{O}(\log^2 n/\log \log n)$
                & Two tier
                    & \cmark
                     \\ \hline
    \panorama \cite{patel2018panorama} 
        & $\mathcal{O}(B)$
            & $\mathcal{O}(\log n \cdot \log \log n)$
                & Cuckoo
                    & \cmark
                     \\ \hline
    \optorama \cite{asharov2020optorama}  
        & $\mathcal{O}(B)$
            & $\mathcal{O}(\log n)$
                & Cuckoo
                    & \cmark
                     \\ \hline \hline
    \rankoram  
        & $\mathcal{O}(n+\sqrt{n}\cdot B)$
            & $4\log n$
                & Permuted array
                    & \cmark
                     \\ \hline 
    \end{tabular}
    \label{tab:hierarchical}
\end{table*}

\section{Related Work}
\label{sec:related}

\paragraph*{Hierarchical ORAMs}
\rankoram belongs to the class of 
\hierarchical schemes that have been extensively studied in the literature.
A \hierarchical distributes the outsourced data across levels that increase exponentially in size.
Each level is implemented as an oblivious hash table.
All the variations of \hierarchical depend on the implementation of this primitive.
The amortized bandwidth overhead of \hierarchical is determined by the cost of the rebuild (offline bandwidth) and the cost of an access (online bandwidth).
In the original proposal by Goldriech and Ostrovsky, to store a database of $n$ blocks of $B$ bits,
at level~$l$, the hash table contains $2^l$ buckets of~$\mathcal{O}(\log n)$ depth.
When accessing a bucket obliviously, a linear scan is performed.
The scheme's (amortized) bandwidth cost of $\mathcal{O}(\log^3 n)$ is dominated by the rebuild phase.

Subsequent improvements were achieved by changing the hashing primitive to an oblivious cuckoo hash table \cite{chan2017oblivious,goodrich2011privacy,goodrich2012privacy,pinkas2010oblivious}.
With cuckoo hashing, the lookup time is constant. 
These schemes incur an amortized $\mathcal{O}(\log^2 n/\log \log n)$ bandwidth cost that is dominated by a rebuilding phase, which relies on expensive oblivious sorting\footnote{Oblivious sorting in $\mathcal{O}(B)$ bits is expensive in practice and is required in many ORAM schemes for the rebuilding phase.
For a discussion on the trade-offs between client memory and bandwidth, see Holland \textit{et al.} \cite{holland2022waksman}.
}.
Patel \textit{et al.}, with \panorama, provide a cuckoo construction algorithm that does not rely on oblivious sorting \cite{patel2018panorama}.
They assume that the input to the construction algorithm is randomly shuffled and the bandwidth overhead is reduced to $\mathcal{O}(\log n\cdot\log\log n)$ blocks.
This idea was extended by Asharov~\textit{et al.}, with~\optorama,  to achieve an optimal,~$\mathcal{O}(\log n)$, bandwidth overhead~\cite{asharov2020optorama}, matching the lower bound of Larsen and Nielsen. \cite{larsen2018yes}.
Note that the lower bound applies to passive ORAMs.
All of the above ORAM protocols operate with $\mathcal{O}(B)$ bits of client memory {and are interactive}.
Table~\ref{tab:hierarchical} summarizes the theoretical performance of these schemes.

\begin{table*}[t]
    \centering
    \caption{
    Comparison of \textit{extended} ORAM schemes with low concrete bandwidth and/or a single round-trip of communication per \access.
    Compared to state-of-the-art approaches, which, as noted in the table, do not use server side computation, \rankoram reduces client memory by a factor of $\mathcal{O}(\log n)$.
    \ringoram is {a variation of recursive} \pathoram that stores additional metadata at the client.
    In addition, \rankoram, \ringoram and \partition require a constant number of round-trips for rebuilds.
    } 
    \begin{tabular}{|c||c|c|c|c|c|} \hline
         & client storage 
            &  bandwidth 
                & single round-trip \access
                    & server comp.
                    \\ \hline
    \sroram \cite{williams2012single}
        & $\mathcal{O}(B)$
            & $\mathcal{O}(\log n)$
                & \cmark
                    & \cmark \\ \hline
    \tworam \cite{garg2016tworam}
        & $\mathcal{O}(\log n)\omega(1) \cdot B$
            & $\mathcal{O}(\log n)$
                & \cmark
                    & \cmark  \\ \hline               
    \bucket \cite{fletcher2015bucket}  
        & $\mathcal{O}(\log n)\omega(1) \cdot B$
            & $\mathcal{O}(\log n)$
                    & \cmark
                        & \cmark \\ \hline                
    \pathoram \cite{stefanov2013path}  
        & $\mathcal{O}(\log n)\omega(1)\cdot B$
            & $16\log n$
                & 
                    & \\ \hline                
    \ringoram \cite{ren2014ring}
        & $\mathcal{O}(n \log n + \sqrt{n}\cdot B)$
            & $3\log n$
                & \cmark
                    &  \\ \hline
    \partition \cite{stefanov2011towards}
        & $\mathcal{O}(n \log n + \sqrt{n}\cdot B)$
            & $3\log n$
                & \cmark
                    &  \\ \hline \hline  
    \rankoram 
        & $\mathcal{O}(n + \sqrt{n}\cdot B)  $         
            & $4 \log n$
                & \cmark 
                    &  \\ \hline                
    \end{tabular}
    \label{tab:practical}
\end{table*}

\paragraph*{Single round-trip}

The execution of an ORAM access depends on \textit{where} the accessed block is located at the server.
Locating the block, with no prior knowledge of where it resides, introduces client-server interaction. 
In small memory, this interactive component can be removed by allowing server-side computation.
For example, \sroram \cite{williams2012single} and \bucket \cite{fletcher2015bucket} place encrypted Bloom filters at the server to separate membership testing from block storage. 
The schemes can then build a layered branching program with paths that depend on the location of the accessed item.
The server queries each Bloom filter and the output is used to unlock the next step in the correct path through the branching program.
The path reveals which blocks to return to the client.
In contrast to our work, which relies on a server performing only read and write requests, the server in the above schemes can perform more complex operations on data that are not supported by common storage devices such as RAM and disk.

\paragraph*{Larger clients and passive servers}

There are many settings (e.g., Intel SGX or cloud setting) in which clients can afford more than $\mathcal{O}(B)$ bits of private memory.
Under this observation, the constructions \partition \cite{stefanov2011towards} and \ringoram \cite{ren2014ring} store metadata concerning server block locations explicitly at the client.
The metadata is stored in \textit{position maps}, occupying $\Theta(n \log n)$ bits using an array, and allows accesses to be executed with a single round-trip of communication.
Both \partition and \ringoram achieve state-of-the-art total bandwidth overhead of $3\log n$, which can be reduced further following optimizations, using server-side computation, proposed by Dautrich \textit{et al}.~\cite{dautrich2014burst}.

These schemes can be adapted to smaller memory environments by recursively storing the metadata in a sequence of smaller ORAMs at the server.
The technique was first presented in  \cite{shi2011oblivious}.
This comes at the cost of increased bandwidth and latency.
The performance of this class of ORAMs is summarized in Table~\ref{tab:practical}.
In comparison to state-of-the-art, \rankoram uses $\mathcal{O}(\log n)$ less memory and obtains a comparable bandwidth overhead.

\paragraph{Compressing metadata}
\label{subsec:compressing}

The position maps for both \partition and \ringoram, with multiple types of metadata per block, occupy $\Theta(n \log n)$ bits.
As $n$ increases, this term begins to dominate the client memory allocation.
To alleviate this burden, Stefanov \textit{et al}.~\cite{stefanov2011towards} compress the position map with a method, which {we call} \counter, that is designed for sequential workloads, but has a worst-case memory allocation of $\mathcal{O}(n \log n)$ bits.
The theoretical properties of different client-side data structures are summarized in Table~\ref{tab:client_ds}.

In contrast to this line work, some ORAMs \cite{cao2021streamline,raoufi2023ab} are designed to reduce space utilization at the server, not the client.
As they focus on a problem orthogonal to ours, we do not include these works in our comparison.
However, it is worth noting that the use of permuted arrays in hierarchical ORAM achieves state-of-the-art space utilization at the server ($\sim50\%$).

\begin{table}[t]
    \centering
    \caption{Comparison of client side data structures {for storing block locations on the server}. Performance of \counter~\cite{stefanov2011towards} is based on our instantiation in Appendix~\ref{sec:counter}.
    \textsf{array} appears in \partition and \ringoram, and can be replaced with \counter.
    Structure \historical encodes information that is \textit{unique} to \hierarchical.
    Note that, unlike \historical, \counter requires server-side computation to operate.
    }
    \begin{tabular}{|c||c|c|} \hline
        & memory
            & update time \\ \hline \hline
    \textsf{array} \cite{chang2016oblivious,dautrich2014burst,ren2014ring,stefanov2011towards} 
        &  $\Theta(n \log n)$
            & $\mathcal{O}(1)$ \\ \hline
    \counter 
        &$\mathcal{O}(n \log n)$
            & $\mathcal{O}(\log^2 n)$\\ \hline \hline
    \historical
        & $\mathcal{O}(n)$
            & $\mathcal{O}(\log^2 n)$\\ \hline    
    \end{tabular}
    \label{tab:client_ds}
\end{table}

\section{Preliminaries}
\label{sec:prelim}

Fixing notation,
we consider the setting where a client outsources~$n$ blocks, each of~$B$ bits, to untrusted storage.

\subsection{Performance Metrics}
\label{sec:perform}
For measuring performance we consider three key parameters: \textit{client memory}, \textit{bandwidth overhead} and \textit{the number of round-trips} (latency).
The size of the client memory measures the amount of storage, both temporary and permanent, required
to execute an ORAM scheme.
The bandwidth overhead refers to the number of blocks, possibly amortized, exchanged between the client and server per virtual access.
It represents the multiplicative overhead of moving from a non-oblivious to an oblivious storage strategy.
The number of round trips counts the rounds of communication between the client and server per virtual access.

\subsection{Definitions}

\paragraph{Security}
We adopt the standard security definition for ORAMs. 
Intuitively, it states that the adversary should not be able to distinguish between two access patterns of the same length.
In other words, the adversary should learn nothing from observing the access pattern.

\begin{definition}[Oblivious RAM \cite{stefanov2011towards}]
Let 
\begin{align*}
    \vec{y} := 
    \{ 
        (\textsf{op}_1, \textsf{a}_1, \textsf{data}_1),
        \ldots,
        (\textsf{op}_m, \textsf{a}_m, \textsf{data}_m)
    \}  
\end{align*}
denote a sequence of length $n$, 
where $\textsf{op}_i$ denotes a $\textsf{read}(\textsf{a}_i)$ or $\textsf{write}(\textsf{a}_i, \textsf{data}_i)$.
Specifically, $\textsf{a}_i$ denotes the logical address being read or written and $\textsf{data}_i$ denotes the data being written.
Let $A(\vec{y})$ denote the (possibly randomized) sequence of accesses to the remote storage given the sequence of data requests $\vec{y}$. 
An ORAM construction is deemed secure if for every two data-request sequences, $\vec{y}$ and~$\vec{z}$, of the same length, their access patterns~$A(\vec{y})$ and~$A(\vec{z})$ are, except by the client, computationally indistinguishable. 
\end{definition}

\paragraph{Oblivious Shuffle}
A key primitive of oblivious RAM solutions, including \rankoram, is oblivious shuffle.
It implements the following functionality.
\begin{definition}[Functionality: Array Shuffle]
Let $\mathcal{P}$ 
denote a set of permutations. On input array $U$, of key-value pairs, and permutation $\pi\in\mathcal{P}$, the Array Shuffle outputs the array $V  = \shuffle(\pi, U)$, where $V[i]= (k,v)$ and $\pi(k)=i$.
\label{def:rand_perm}
\end{definition}

{We assume that the permutation function,~$\pi$, is given to the algorithm in a form that allows for its efficient evaluation. For example, it could be provided as a seed to a pseudo-random permutation.}
Oblivious algorithms preserve the input-output behaviour of a functionality and produce an access pattern that is independent of the input.
We now define the notion of oblivious algorithm.
\begin{definition}[Oblivious Algorithm]
Let $A(M^\mathcal{F}(x))$ denote the access pattern produced by an algorithm $M$ implementing the functionality $\mathcal{F}$ on input $x$.
The algorithm $M$ is oblivious if, for every two distinct inputs, $x_1$ and $x_2$, of the same length, except to the client, their access patterns, $A(M^\mathcal{F}(x_1))$ and $A(M^\mathcal{F}(x_2))$, respectively, are computationally indistinguishable.
\label{def:oblv_perm}
\end{definition}
An oblivious shuffle implements functionality Definition \ref{def:rand_perm} and does not reveal anything
about the input permutation through its access pattern.
This is because two executions on any distinct pair of input permutations should be indistinguishable.

\subsection{\hierarchical}
\label{subsec:hierarchical}

The Hierarchical ORAM contains a hierarchy of oblivious hash tables $T_0, \ldots, T_L$, with $L = \log n$. 
In the words of Goldriech and Ostrovsky \cite{goldreich1996software}, the ORAM consists of ``a hierarchy of buffers of different sizes, where essentially we are going to access and shuffle buffers with frequency inversely proportional to their sizes''.
The \textit{hash table} abstraction contains a look-up query and a construction algorithm.
For the construction algorithm to be oblivious, by Definition~\ref{def:oblv_perm}, the input blocks must be placed in the table without leaking their locations through the access pattern.

A \hierarchical has the following structure.
The hash table~$T_l$ stores~$2^l$ data blocks.
Next to each table, a flag is stored to indicate whether the hash table is \textit{full} or \textit{empty}.
When receiving a request to an address,~$x$, the ORAM operation involves both an \access and \rebuild phase:
\begin{enumerate}
    \item \access: Access all non-empty hash tables in order and perform a lookup for address $x$. If the item is found in some level $l$, perform dummy look ups in the non-empty tables of $T_{l+1},\ldots, T_L$.
    If the operation is a $\readop$, then store the found data and place the block in $T_0$.
    If the operation is a $\writeop$, ignore the associated data and update the block with the fresh value.
    \item \rebuild: Find the smallest empty hash table $T_l$ (if no such level exists, then set $l=L$). 
    Merge the accessed item and all of $\{T_j\}_{j\leq l}$ into $T_l$.
    Mark levels $T_0, \ldots, T_{l-1}$ as empty.
\end{enumerate}
\noindent
A block is never accessed twice in the \textit{same}
hash table in between rebuilds at a given level.
This invariant is crucial to the security of the scheme.
As each block is always retrieved from a different location, the sequence of hash table probes produced by \hierarchical appears random to an adversary.

{Note that \access is interactive since the client does not know which level a block belongs to. 
That is, the client has to query the levels sequentially until the target block is found.}
This requires a round trip per level and increases  the latency of the protocol. 
{Our client-side data structure in \rankoram is designed to remove this cost}.

\subsection{Indexed Dictionary}
\label{sec:indexed}

Our client-side data structure is built on set-membership structures that support the following operations on the set $S$
\begin{align}
  S.\indll(r) &= \text{return the }  r^{\text{th}} \text{ smallest item in } S \label{eqn:index} \\
  S.\rank(x) &= |\{y \mid y \in S, y< x \}|
 \label{eqn:rank_phi}
\end{align}
A data structure that supports these operations, in addition to a \membership query ($x  \stackrel{?}{\in} S$), is called an \textit{indexed dictionary} \cite{raman2007succinct}.
For example, for the set $S = \{2,5,7,9\}$, the query functions evaluate as $S.\rank(7)=3$ and~$S.\indll(3)=7$.

\section{\rankoram Data Structures}
\label{sec:overview}

\rankoram is an ORAM solution that addresses the limitations preventing the adoption of Hierarchical ORAMs in practice.
It differs from prior instantiations of \hierarchical  by storing block metadata at the client. 
This results in two crucial improvements: first, we reduce the number of round-trips of communication per \access; and second, we can store blocks at the server in a simple and efficient data structure.
The latter improvement leads to a bandwidth and latency efficient \rebuild phase.

At a high level, \rankoram follows the template of \hierarchical (Section~\ref{subsec:hierarchical}) with the addition of a client-side data structure. 
\rankoram builds on the observation that, in order to achieve a single round-trip, the client needs to know \textit{where} a block is located, a priori, to avoid searching for it interactively at the server. 
To this end, we design a compact data structure called \historical to encode this information at the client. 
In particular,
\historical computes, for each block, its current level and its hash table position within the level. 
Subsequently, courtesy of this metadata, \historical informs the protocol of where to retrieve an accessed block and of where to retrieve dummies at the server.
Thus, all the level probes can be batched in a single non-interactive request to the server. A naive way of instantiating \historical requires $\Theta(n \log n)$ bits. Instead we design a succinct data structure that requires only $O(n)$ bits and supports efficient update and lookup operations.

In addition, \historical allows us to improve the design of server-level storage and instantiate each level with an efficient permuted array.
This allows \rankoram to dispense with the (expensive) cuckoo hash table.
Hash tables are needed in \hierarchical to enable the efficient retrieval of elements in the domain $[1,n]$ from a table of size smaller than $n$ (that is, size $2^l$ at level~$l$). 
\historical compactly stores this mapping from the block domain to the hash table indices.
Therefore, a permuted array suffices as the level hash table. 
{Rebuilding a permuted array relies on a single call to an oblivious shuffle, making the \rebuild phase not only simpler than constructing a cuckoo hash table, but also more efficient.

The remainder of this section details the two core data structures of \rankoram.
We outline our client-side data structure \historical, which enables the performance improvements enjoyed by \rankoram.
Afterwards, we describe our
oblivious hash table implementation.
In the subsequent Section we describe how these data structures are used during the \access and \rebuild phases.

\subsection{Client Data Structure}

\rankoram combines \hierarchical with a client-side data structure \historical.
Let $S_l = \{ a \mid (a,\square)\in T_l\}$
denote the set of logical addresses at level~$l$ in a \hierarchical. 
\historical maintains each set $S_l$, for $l\in [L-1]$, in a compressed indexed dictionary (as defined in Section \ref{sec:indexed}).
As access patterns typically have low entropy, for example, access patterns in file systems are highly sequential~\cite{oprea2007integrity},
the sets are compressible.

\historical, comprised of the collection of indexed dictionaries~$\mathcal{S}=\{S_0, S_1,\ldots, S_{L-1}\}$, \sloppy supports the following two functions:
\begin{align}
    \level(x, \mathcal{S}) =  \min \{ i \mid (x,\square) \in T_i, i \in \{1, \ldots, L\}\}, \label{eqn:level}\\
    \posll(x, \mathcal{S}) = 
    \begin{cases} 
        S_{\level(x)}.\rank(x) &\mbox{if } \textsf{level}(x) < L \\
        x & \mbox{if }\level(x) = L 
    \end{cases}.
    \label{eqn:position}
\end{align}
The function \level denotes the level a block belongs to
and is calculated in a sequence of membership queries starting from $S_0$ and progressing through the hierarchy.
The function $\posll(x)$ denotes the \rank of address $x$ \textit{within} its current level. 
The \level functionality is observed in prior work \cite{ren2014ring,stefanov2011towards}, where it is implemented in an array (called the position map) mapping element addresses to a collection of auxiliary information.
The \posll functionality is novel to this work and is used to map addresses into hash table locations.
An example of the \posll function for a given instance of \rankoram is given in Figure \ref{fig:acc}.

When a \rebuild happens at the server ($T_{l}\gets \bigcup_{i=1}^{l-1} T_{i}$), \historical is updated accordingly: $S_l \gets \bigcup_{i=1}^{l-1} S_{l}$ and $S_i \gets \varnothing, \forall i \in \{1,\ldots, l-1\}$.
Thus, one requirement for the choice of indexed dictionary is that it supports efficient merging.
Further, as each element belongs to exactly one level,
we do not need to store~$S_L$.
To evaluate Equation \eqref{eqn:level}, if the element is not a member of $\{S_l\}_{l \leq L-1}$, it must be a member of $S_L$.
Further, 
we do not need \rank information at level $L$.
This is covered in more detail in Subsection~\ref{sec:hash_table}.

The \level information is a function of block recency.
Consequently, \historical can be used to estimate the recency of a block.
This technique is similar to that used by Holland \textit{et al.} to estimate recency in small memory and with an arbitrary amount of relative error \cite{holland2020recency}.
Our work differs in the low-level details, such as the implementation of the component dictionaries, and with the additional requirement that we need to support \rank queries at each level.

To instantiate \historical, any indexed dictionary can be used.
For our result, we use run-length encoding \cite{bradley1969optimizing} and achieve a memory allocation of $\mathcal{O}(n)$ bits.
This is the worst-case allocation and run-length encoding will compress any entropy in the access trace.
A naive approach, of storing the \level and \posll information in an array, would require $\mathcal{O}(n \log n)$ bits.
Notably, run-length codes can be merged efficiently while keeping stored information in a compressed state.
The technical details are provided in Section~\ref{SEC:CLIENT}.

\historical can be used by existing \hierarchical solutions that employ other hash tables to reduce round-trips of communication.
{However, we note that \historical cannot be used to encode the metadata arrays for \ringoram and \partition, as these arrays do not exhibit temporal locality. 
For example, both constructions store arrays mapping addresses to random numbers.
These numbers not only contain high entropy, but are fixed for a given address between accesses.
In contrast, the level data stored in \historical is a function of block recency and changes with each \rebuild.

\subsection{Permuted Array as Oblivious Hash Table}
\label{sec:hash_table}

\rankoram implements the level-wise hash table at the server using a permuted array.
Fixing notation, let $\mathbf{T}_l$ refer to the array, located at the server, that stores the elements of level $l$ and let $n_l:=2^l$.
The array has length $2n_l$ and stores at most $n_l$ real elements and, consequently, at least $n_l$ dummy elements.
For levels $l<L$, we utilize the \rank information to assign each element in level $l$ to a \textit{unique} index in the domain $[n_l]$.
The \rank information is available through the indexed dictionaries that form \historical.
Elements are then permuted, according to the permutation $\pi_l:[2n_l]\rightarrow[2n_l]$, using their ranks:
\begin{align}
    \mathbf{T}_l[\pi_l( S_l.\rank(a))] = (a,\textsf{data}).
    \label{eqn:array_access}
\end{align}
The \rank can be interpreted as the unpermuted index of the item. 
For level $L$, we do not need to worry about mapping the address space onto a smaller domain and we can proceed by simply permuting the address space: 
\begin{align}
    \mathbf{T}_L[\pi_L(a)] = (a,\textsf{data}).
    \label{eqn:array_access_L}
\end{align}
Thus, each block can be retrieved with the \posll function.
When a level is rebuilt, a new permutation function is generated with fresh randomness\footnote{Stefanov \textit{et al.} provide a collection of pseudorandom permutation functions that are fast and suitable for small domains \cite{stefanov2012fastprp}}.

\partition places blocks in $\sqrt{n}$ small Hierarchical ORAMs each of size $\sim \sqrt{n}$.
There are some similarities between our construction and the component ORAMs of the \partition, which also utilize permuted arrays.
The key difference between our setup and that of partitioning is the size of the levels.
Partitioning produces smaller levels.
This allows shuffling to be performed exclusively at the client (the whole level is downloaded and permuted locally).
In contrast, as the largest level of \rankoram has size $n$, we require an interactive shuffling algorithm.
The advantage of our approach is that, with the rank information compactly encoded in $S$, 
we can access blocks at the server directly through $\pi_l$ (see Equation \eqref{eqn:array_access}).
In other words, we are able to map the set $S_l$ onto the domain of $\pi_l$ without collisions.
Without rank information, \partition is required to store the offsets explicitly.
Through our experiments, we demonstrate that our approach leads to significant savings in client memory.

\section{\rankoram Operations}
\label{sec:rank}

We now describe the \rankoram implementation of the \hierarchical template.
This covers the \access and \rebuild operations.
In addition, we detail the implementation of an \evict operation that is used within a \rebuild to remove untouched blocks from expired hash tables.

\begin{figure}[t]
    \centering
    \includegraphics[width=0.6\linewidth]{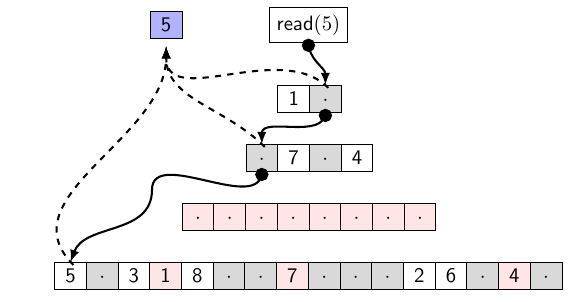}
    \caption{
    A $\readop(5)$ in \rankoram.
    Level 2 (in red) is empty.
    The \posll function, supported by \historical, is used to locate an item in a level.
    For example, $\posll(4)=S_1.\textsf{rank}(4)=0$.
    Therefore, block $4$ is located at index $\pi_1(0)=3$ in level $1$.
    The \readop operation queries the sequence of hash tables (a permuted array padded with dummies) looking for item~$5$.
    The server returns dummy elements at levels~0 and~1.
    The queried item is found and returned to the client at level~$3$.
    After the payload of the block is read, a \rebuild occurs.
    Table~$T_2$ is the smallest empty hash table.
    We then (obliviously) build $T_2$ on the input $5||T_0||T_1$  
    and mark~$T_0$ and~$T_1$ as empty.}
    \label{fig:acc}
\end{figure}
\begin{algorithm}[t]
    \SetAlgoLined
    \DontPrintSemicolon
    \SetKwProg{myproc}{define}{}{}
    \myproc{\textup{$\textsf{access}(a, \textsf{op}, \textsf{data}^{\prime})$}}
    {
        $l \gets \level(a, \mathcal{S})$\tcp*{retrieve metadata } 
        $r \gets \posll(a, \mathcal{S})$\;
        initialize empty query $Q$
        \tcp*{batch query to send to the server}
        \For{$i \in \textup{\textsf{occupiedLevels}}$}
        {
            \If{$i = l$}
            {
                $Q \gets Q \cup (i, \pi_l(r))$\;
            }
            \Else
            {
                $Q \gets Q \cup (i,\pi_i(\textsf{dummyCntr}_i))$\;
                $\textsf{dummyCntr}_i \gets \textsf{dummyCntr}_i+1$\;
            }
        }
        send $Q$ to the server\;
        \colorbox{red!25}{
        initialise empty output $O$}\tcp*{server executes $Q$}
        \For{$i \in \textup{\textsf{occupiedLevels}}$}
        {
            \colorbox{red!25}{retrieve $(i, \textsf{index}) $ from $Q$}\;
            \colorbox{red!25}{$O \gets O \cup \mathbf{T}_i[\textsf{index}]$}\;
        }
        \colorbox{red!25}{send $O$ to the client}\tcp*{returns result to client}
        $(a, \textsf{data}) \gets$ retrieve and decrypt block $a$ from $O$ 
        \If{\textsf{op} = \textsf{write}}
        {
            $\textsf{data}\gets \textsf{data}^{\prime} $
        }
        \colorbox{red!25}{$\mathbf{T}_0 \gets \textsf{encrypt}(a, \textsf{data})$}\;
        $S_0 \gets a$\;
        $\textsf{count} \gets \textsf{count}+1 \bmod 2^L$\;
        \rebuild()
        \KwRet \textsf{data}\;
    }
\caption{Rank ORAM access. Lines highlighted in red represent operations performed at the server.
}
\label{alg:access}
\end{algorithm}
\subsection{Access Phase}

The \access algorithm, formalized in Algorithm \ref{alg:access}, is similar to that of prior Hierarchical ORAMs. 
The primary difference is that the procedure begins by determining 
the block metadata (the \level and \posll of the data block) {stored at the client}. 
The \posll function is used {to find} the index of the block within its level (line 7).
With the \level information, we know, a priori, {which level the block is stored at and} which levels receive a dummy probe.
Thus, membership queries to all hash tables can be accumulated, as a set containing one hash table index per non-empty level, and sent to the server as a single request (lines 6--10).
The server retrieves all requested blocks in one batch, resulting in one round-trip of communication. 

A dummy counter is used to ensure that the scheme only returns untouched blocks from  a given hash table instance.
Recall that this constraint is necessary for the obliviousness of the scheme.
When a dummy is retrieved (line 9), we increment the counter (Line 10) so that an untouched dummy is retrieved at the next dummy request.
An example of the \access procedure is presented in Figure~\ref{fig:acc}.

\subsection{Rebuild Phase}
A \rebuild enforces the invariant that a block is never retrieved more than once from a given hash table instance.
This ensures that the sequence of hash table probes appears random to the adversary.
The \rebuild carries out the invariant by periodically moving accessed blocks up the hierarchy into fresh hash table instances.
Consequently, a block at level $l<L$ has recency less than $2\cdot 2^l$.

Given that the hash tables are implemented as permuted arrays, the rebuild function is straightforward. 
We update \historical and subsequently rearrange the server's memory. 
For a \rebuild into level $l$, we merge the compressed dictionaries in levels~$\{0,\ldots,(l-1)\}$ to obtain the dictionary $S_l$.
The rebuild at the server involves a single oblivious shuffle. 
The input array is the concatenation of the untouched blocks (including dummies) in levels $\{0,\ldots,(l-1)\}$.
An \evict operation (Algorithm \ref{alg:evict}) is executed to construct the array of untouched blocks.
The input array is padded with dummy blocks\footnote{The dummy blocks are indexed so that their encryptions are indistinguishable from real blocks.} to the width of the output array.
We generate a new pseudo-random permutation, $\pi_l$, (Line 2) and 
the input permutation for the oblivious shuffle is the composition of functions $\pi_l\circ \textsf{rank}(\cdot, S_l)$.
Any oblivious shuffle algorithm can be used to perform this step~\cite{ohrimenko2014melbourne,patel2017cacheshuffle}.

The procedure concludes by updating some client-side data; the dictionaries at levels~$\{0,\ldots,(l-1)\}$ are deleted; the dummy counter is set for level $l$; and the set of non-empty (or occupied) levels is adjusted (line 10). 
The \rebuild procedure is presented in Algorithm~\ref{alg:rebuild}.
The function \textsf{msb} (Line 2) computes the most significant bit of the input.
\begin{algorithm}[t]
    \SetAlgoLined
    \DontPrintSemicolon
    \SetKwFunction{ind}{input\_switch}
    \SetKwProg{internal}{Internal State}{}{}
    \SetKwProg{myproc}{define}{}{}
    \myproc{\textup{$\textsf{rebuild}()$}}
    {
        $l \gets \textsf{msb}(\textsf{count})+1$\;
        $\pi_l \gets$  a pseudo random permutation on the domain $[2^l]$\;
        $S_l \gets S_0 \cup \cdots \cup S_{l-1}$\;
        \colorbox{red!25}{$\mathbf{I} \gets  \evict(0) \mid\mid \cdots \mid\mid \evict(l-1)$}\;
        \colorbox{red!25}{$\mathbf{I} \gets  \mathbf{I} \mid \mid [n_l$ dummy blocks]}\;
        \colorbox{red!25}{$\mathbf{T}_l \gets  \shuffle(\pi_l(\textsf{rank}(\cdot, S_l)), \mathbf{I})$}\;
        \For{$i\in\{0, \ldots, l-1\}$}
        {
            $S_i \gets \varnothing$\;
        }
        $\textsf{dummyCntr}_l \gets |S_l|$\;
        $\textsf{occupiedLevels} \gets \{l\} \cup \textsf{occupiedLevels} \setminus \{0, \ldots, l-1\} $\;
        \KwRet\;
    }
\caption{Rank ORAM Rebuild
}
\label{alg:rebuild}
\end{algorithm}

\rankoram can use any prior oblivious shuffle during a \rebuild.
However, in our context, we have dummy blocks in both the input and output arrays.
A given dummy block can be placed in \textit{any} vacant location in the output array and we can exploit this fact to gain performance improvements relative to general shuffling algorithms.
Consequently, we construct a variation of the \cacheshuff~\cite{patel2017cacheshuffle}, named \shortqueue, that leverages a shuffling instance where dummies are placed in the output array.
For comparison, with the standard \cacheshuff, permuting level $l$ would cost $9 \cdot 2^l$ blocks of bandwidth.
With \shortqueue, we reduce this cost to $7\cdot 2^l$.
The details of \shortqueue are presented in Appendix \ref{sec:shuffling} and its performance is summarized below.

\begin{definition}[Dummy shuffle functionality]
On an input array of length $n$ and a permutation function $\pi:[2n]\rightarrow [2n]$, the dummy shuffle functionality outputs an array of length $2n$ with the input elements placed according to $\pi$ and the remaining locations filled with dummies.
\label{def:dummy_shuffle}
\end{definition}

\begin{restatable}{lemma}{lemshortq}
\textup{\shortqueue} is an oblivious dummy shuffling algorithm.
On an input of length $n$, \textup{\shortqueue} completes in $7n$ blocks of bandwidth, requires $\sqrt{n}$ round-trips of communication and uses $\mathcal{O}(\sqrt{n}\cdot B)$ bits of private memory.
\label{LEM:SHORTQUEUE}
\end{restatable}

\subsection{Eviction procedure}
Hash table eviction (used in \rebuild phase) involves removing all untouched blocks from the hash table in order from lowest index to highest index.
This procedure can be carried out by the server.
For completeness, and as we do not allow additional server power, we provide two efficient ways for the client to perform eviction.
We could store a bitmap, locally at the client, that indicates the untouched indices. 
This approach requires an additional $4n$ bits but does note impact our asymptotic result.
Otherwise, we can use the inverse permutation function, $\pi_l^{-1}$ for level $l$, to enumerate the ranks of the elements in the order in which the ranks appear in the array.
The rank can be used to determine if the corresponding element is real or dummy; an element with a \rank larger than the cardinality of the level, $|S_l|$, is a dummy block.
We can then determine if it is a touched block.
If it is real and also belongs to a lower level, then it is touched.
If it is a dummy and the rank is lower than $\textsf{dummyCntr}$, then it is touched.
We skip over touched elements and only retrieve untouched elements.
By assumption the inverse permutation is pseudo-random and can be stored and evaluated efficiently.
{Finally, it is worth noting that we
could avoid eviction by requesting the server to delete every block retrieved in Line~14 in Algorithm~\ref{alg:access}.}

\section{Historical Membership with Run-length Encoding}
\label{SEC:CLIENT}

We now provide an encoding for the indexed dictionaries that are the components of \historical.
In the ORAM setting the encoding must obtain an efficient worst-case compression.
This allows the client to safely allocate, a priori, a compact allotment of memory to the data structure.
Otherwise, the adversary could use the size of the memory allocation at the client as another side-channel to infer that the sequence of accesses belongs to a subset of the possible access patterns. 
An alternative method for compressing metadata, proposed by Stefanov \textit{et al.} \cite{stefanov2011towards} (See Appendix~\ref{sec:counter}), has poor worst-case behaviour and demonstrates that achieving this property is non-trivial.
Our method relies on an intricate use of run-length encoded strings  \cite{bradley1969optimizing} and a search procedure over them.  
This approach achieves a $\mathcal{O}(n)$ bit memory allocation in the worst-case, encoding the requisite metadata with only a constant number of bits per block.

\begin{algorithm}[t]
    \SetAlgoLined
    \DontPrintSemicolon
    \SetKwInOut{internal}{Internal state}{}{}
    \SetKwProg{myproc}{define}{}{}
        \myproc{\textup{$\evict(l)$}}
    {
        \colorbox{red!25}{$ \mathbf{E} \gets []$}\tcp*{empty array of length $n_l$}
        \For{$i \in \{0,1,\ldots, 2\cdot n_l-1\}$}
        {
        $r \gets \pi_l^{-1}(i)$\;
        \If{$r< |S_l|$}
        {
            \tcp{$r$ represents a real element}
            $a \gets \indll(r, S_l)$\;
            $l^{\prime} \gets \textsf{level}(a, \mathcal{S})$\;
            \If{$l^{\prime} = l$}
            {
                \tcp{$a$ is untouched}
                \colorbox{red!25}{$\mathbf{E} \gets \mathbf{E} \mid\mid \mathbf{T}[i]$}\;
            }
        } 
        \ElseIf{$r \geq \textup{\textsf{dummyCntr}}$}
        {
            \tcp{$r$ represents an untouched dummy index}
            \colorbox{red!25}{$\mathbf{E} \gets \mathbf{E} \mid\mid \mathbf{T}[i]$}\;
        }
        }
       
        \KwRet $\mathbf{E}$\;
    }
\caption{Rank ORAM hash table eviction: the procedure removes all blocks that were not touched during $\access$ at level $l$}
\label{alg:evict}
\end{algorithm}

\subsection{Run-length encoding}
\label{subsec:encoding}
A run-length encoding \cite{bradley1969optimizing} compiles the set $S=\{x_1, x_2, \ldots, x_{n_l}\}$, where $x_i < x_j$ for all pairs $i<j$, as the string:
\[
\textsf{rle}(S) = x_1 \circ (x_2-x_1) \circ (x_3-x_2) \circ \cdots \circ (x_{n_l}-x_{n_l-1}) \circ (x_{n_l+1}-x_{n_l})\,,
\]
where $x_{n_l+1}=n$.
Each sub-string can be encoded with a prefix-free Elias code~\cite{elias1975universal}. 
This encodes integer~$x$ in $\mathcal{O}(\log x)$ bits, so a level-$l$ dictionary requires:
\begin{align}
    \sum_{i=2}^{n_l+1} \mathcal{O}(\log(x_i - x_{i-1})) &\leq \sum_{i=2}^{n_l+1} \mathcal{O}(\log n/n_l)) &&   \nonumber\\
    &= \mathcal{O}(n_l\log(n/n_l)) \text{ bits,} \label{eqn:rle}
\end{align}
where
the first inequality holds as $\log$ is a convex function.
There are dictionaries that are more (space) efficient than this.
However, run-length codes have the advantage of being easily mergeable: given encodings of the sets $S_1$ and $S_2$, one can enumerate $S_1 \cup S_2$, \textit{in order}, in $\mathcal{O}(|S_1|+|S_2|)$ time and with a working space of $\mathcal{O}(1)$ words.
For example, to merge $\textsf{rle}(S_1)$ and $\textsf{rle}(S_2)$, we extract, iteratively, the smallest key from the front of its corresponding code and place it in the new code.
This means that the keys in  $S_1 \cup S_2$ are enumerated in order, without uncompressing the codes.

\subsection{Auxiliary data structure}
\label{subsec:auxiliary_ds}
The efficiency of the \level (Equation \eqref{eqn:level}) and \posll (Equation \eqref{eqn:position}) functions depend, respectively, on the efficiencies of the \membership and \rank queries on the component dictionaries.
To support these functions efficiently, we supplement each code with an auxiliary structure of \textit{forward pointers} that is constructed as follows.
We divide the run-length encoding $\textsf{rle}(S_l)$ into $\mathcal{O}(|S_l|/\log n)$ segments of equal cardinality of order $\Theta(\log n)$.
The auxiliary structure consists of an array of addresses $A$, where $A[i]$ denotes the virtual address stored at the beginning of segment $i$, and an array of pointers $P$, where $P[i]$ points to $A[i]$ in the run-length code.
Notably, $A$ is utilized, in \membership and \rank, for the fast identification of the correct segment.
Combined, these arrays support forward access to the beginning of each segment.

\paragraph{Query algorithms}
Let $W=\Theta(\log n)$ denote the width of every segment.
To execute a fast $\membership(a,S)$ query, the procedure performs a binary search on $A$ to realize the correct segment, that is, the index $i$ such that $a\in [A[i], A[i+1])$.
With $P[i]$, the procedure then jumps to the correct segment of the run-length code in a single probe.
The query is completed with a linear scan of the segment.
Pseudo code is provided in Algorithm \ref{alg:rle_query} in Appendix~\ref{sec:alg_5}.
The query $\rank(a,S)$ can be decomposed into two parts:
\[
\rank(a,S) = \rank(A[i],S) + \rank(a,S[i])\,,
\]
where $S[i]$ is the segment that contains the neighborhood of $a$.
As above, the index~$i$ is determined through a binary search on $A$.
As the segments have equal cardinality, $\rank(A[i],S) = i\cdot  W$.
The local query on the segment $S[i]$ is computed through a linear scan.
The $\indll(i,S)$ query is computed as follows.
The procedure first identifies the correct segment as $j= \lfloor i/W \rfloor$.
It then performs a local \indll query on segment $j$ for input $i^{\prime} = i - j\cdot W$.
This can be executed through a linear scan of the segment.
The run-time cost of the algorithms are summarized in the following.
\begin{lemma}
    \textup{\membership}, \textup{\rank} and \textup{\indll} queries on $\textup{\textsf{rle}}(S_l)$ take $\mathcal{O}(\log n)$ time.
    \label{lem:rle_query_times}
\end{lemma}
\begin{proof}
The code is divided into segments of width $\mathcal{O}(\log n)$.
The \rank and \membership operations involve two stages; (1) locate the correct segment in the code (as specified by the auxiliary structure); 
(2) scan the segment in search of the block address.
The auxiliary array $A$ contains an ordered sequence of $\mathcal{O}(|S_l|/ \log n)$ block addresses that correspond to the starting address of each segment.
Thus the forward pointer to the correct segment can be located, through a binary search on $A$, in $o(\log |S_l|)$ time.
The subsequent linear scan can be executed in time proportional to the length of the segment.
Thus, both \rank and \membership require $\mathcal{O}(\log n)$ time to complete.
The \indll query does not require binary search on the auxiliary structure and can locate the correct segment in constant time. 
It also requires a linear scan of the segment and completes in $\mathcal{O}(\log n)$ time.
\end{proof}

\subsection{Performance of \historical}

With the auxiliary structure established, we conclude by evaluating the performance of \historical.
We begin with the memory allocation.

\begin{lemma}
\textup{\historical}, compressed with run-length codes, requires $\mathcal{O}(n)$ bits of memory.
\label{lem:hist_mem_app}
\end{lemma}

\begin{proof}
With respect to space- and time-efficiency, the worst-case occurs when the levels are full, that is, when $|S_l| = 2^l, \forall l \in [L]$. 
At level $l$ each array in the auxiliary structure has $|S_l|/\Theta(\log n)$ entries of $\mathcal{O}(\log n)$ bits.
The auxiliary structure therefore occupies $\mathcal{O}(|S_l|)$ bits.
Thus, the set $\{\textsf{rle}(S_l)\}_{l\in[l]}$, following Inequality (\ref{eqn:rle}), for some constant $c>0$, uses
\begin{align*}
    \sum_{l=1}^{L-1} c |S_l| \log (n/|S_l|) + \mathcal{O}(|S_l|)  &\leq c\sum_{l=1}^{L-1} 2^l \cdot \log (n/2^l) + \sum_{l=1}^{L-1} \mathcal{O}(2^l) \\
    &= c \left[\sum_{l=1}^{L} 2^l\cdot (L - l)\right] + \mathcal{O}(n)  \\
    &=c \left[\sum_{l=0}^{L-1} 2^l + \sum_{l=0}^{L-2}2^l + \cdots + \sum_{l=0}^0 2^l\right] + \\
    & \quad\mathcal{O}(n) \\
    &\leq c\cdot 2^{L+1} +  \mathcal{O}(n)\\
    &= \mathcal{O}(n) \text{ bits.}
\end{align*}
\end{proof}
The update time is dominated by the cost of merging the run-length codes.
The choice of code is important:
merging should be efficient and  avoid decompressing the codes in a manner that would trespass over the memory bound. 
\begin{lemma}
On a database of size $n$, \textup{\historical} supports updates in $\mathcal{O}(\log^2 n)$ amortized time.
\label{lem:histmem_update_app}
\end{lemma}

\begin{proof}
Following the schedule of \hierarchical, each update involves inserting an address into $S_0$ and performing a merge.
For a merge into level $l$, in order to stay within our memory bounds, we extract the addresses from $S_0 \cup S_1 \cup \cdots \cup S_{l-1}$, in sorted order, and build the run-length code, on the fly, without resorting to a plain form representation.
To do this, we retrieve one address from the front of each code and place it in a list of length $l$.
Then, for $|S_l|$ rounds, we retrieve the smallest address from the list (deleting any duplicates) and place it in the next position of $\textsf{rle}(S_l)$ by computing the corresponding run-length.
When an address is removed from the list, we retrieve the next address from the front of the corresponding code by adding the next run-length to the removed address.

The list has length $l$ and we perform one scan per round.
Inserting a new run-length into $S_l$ and extracting an address from a code both take constant time.
Therefore, a merge takes $\mathcal{O}(l \cdot |S_l|)$ time.

The auxiliary structure can be built with one scan of the code.
For $W=\Theta(\log n)$, we read $W$ run-lengths at a time.
When we reach the beginning of each segment, we place both the starting address and offset in the corresponding auxiliary arrays.
This process completes in $|S_l|$ time.

As a merge occurs every $|S_l|$ updates, the amortized cost of merging into level $l$ is $\mathcal{O}(l)$.
Aggregated across all levels $\{0,1,\ldots,L-1\}$ (we do not perform a merge into level $S_L$), this leads to an amortized update cost of 
\begin{align*}
    \sum_{l=0}^{L-1} \mathcal{O}(l) &= \mathcal{O}(L^2) = \mathcal{O}(\log^2 n)\,.
\end{align*}
\end{proof}

The query functions are listed in Equations \eqref{eqn:level} and \eqref{eqn:position}.
Their runtime costs are a function of the cost of evaluating the \membership and \rank primitives on the component dictionaries. 
\begin{lemma}
On a database of size $n$, \textup{\historical} admits \textup{\posll} and \textup{\level} queries in $\mathcal{O}(\log^2 n)$ time.
\label{lem:histmem_query_app}
\end{lemma}

\begin{proof}

The $\level(x)$ function sequentially probes, from bottom to top, the dictionaries $S = \{S_0, S_1, \ldots, S_L\}$ with \membership queries.
In the worst-case, that is, when $x \in S_L$, $L = \mathcal{O}(\log n)$ \membership queries are preformed.
As each query takes $\mathcal{O}(\log n)$ time, by Lemma~\ref{lem:rle_query_times}, the $\level$ function is evaluated in $\mathcal{O}(\log^2 n)$ time in the worst-case.
The $\posll(x)$ function requires one additional \rank query on $S_{\level(x)}$ at a cost of $\mathcal{O}(\log n)$ time.
\end{proof}


\section{\rankoram Performance and Security}
\label{sec:security}

In this section we first describe two optimizations to reduce both the online and offline bandwidth of \rankoram.
We then evaluate the overall performance of the protocol.

\subsection{Optimizations}
\label{subsec:opt}
We begin with a modification to the hierarchical ORAM template (see Section \ref{subsec:hierarchical}).
As the oblivious shuffle requires $\mathcal{O}(\sqrt{n}B)$ bits of memory, we can afford to store a collection of the smaller levels at the client {(a common optimization found in \partition and \ringoram)}.
To stay within the memory bound, we trim the server side structure and store levels $1$ to $L/2$ at the client.
This reduces the \access bandwidth by a half.

In addition, to save bandwidth, the \evict and \shuffle subroutines of the \rebuild can be intertwined.
$\evict(0) \mid\mid \cdots \mid\mid \evict(l-1)$ and retrieve untouched blocks or generate fresh dummies as the \shuffle algorithm requires them.

\subsection{Performance}

\begin{lemma}[Round-trips]
    In \textup{\rankoram}, an \textup{\access} requires one round-trip of communication and a \textup{\rebuild} requires an amortized constant number of round-trips of communication.
    \label{lem:one-round}
\end{lemma}

\begin{proof}
    Algorithm \ref{alg:access} details an \access. 
    There is one transmission of hash table probes in Line 12 and one transmission of the required blocks in Line 16.
    This represents a single round-trip of communication.

    For \rebuild, first note that the for-loop of the \evict operation can be executed in a sequence of batches of size $\sqrt{n}$.
    We generate $\sqrt{n}$ untouched indicies at the client and download the corresponding blocks in a single-round of communication (the blocks can be subsequently written to the server in a single batch).
    Similarly, By Lemma \ref{LEM:SHORTQUEUE}, the \shortqueue requires $\mathcal{O}(\sqrt{n})$ round-trips of communication.
    As a \rebuild at level $l$ occurs every $2^l$ updates, each level completes its \rebuild in an amortized $\mathcal{O}(\sqrt{2^l}/2^l)= 2^{-l/2}$
    round-trips of communication. 
    Summing across all $\log n$ levels, the total cost of rebuilding \rankoram is (amortized)
    \[
    \sum_{l=0}^{\log n} 2^{-l/2} = \mathcal{O}(1)
    \]
    round trips of communication.
\end{proof}

\begin{lemma}[Bandwidth]
The amortized bandwidth overhead for \textup{\rankoram} is $4\log n$.
\label{lem:concrete_bandwidth}
\end{lemma}
\begin{proof}

We reduce bandwidth by storing levels $1$ to $L/2$ at the client (see Subsection \ref{subsec:opt}).
Therefore, for online bandwidth, \rankoram downloads, $L/2=\log n/2$ blocks from the server per \access.
Similarly, for offline bandwidth, we only need to account for the amortized cost of rebuilding levels $(L/2+1)$ to $L$.
For $l>(L/2+1)$, a level \rebuild costs $7n_l$ blocks of bandwidth by Lemma \ref{LEM:SHORTQUEUE}.
Further, a \rebuild occurs every $n_l$ updates. 
Thus, the amortized bandwidth cost of maintaining a level stored at the server is $7$ blocks.
As there are $L/2$ levels stored at the server, the amortized offline bandwidth overhead is $7 \cdot L/2 = 3.5 \log n$. 
Therefore, total bandwidth is $4 \log n$. 
\end{proof}

\begin{lemma}[{Client total work}]
    Per {each combined \textup{\access} and \textup{\rebuild}} operation, the cost incurred by \textup{\historical} is amortized {$\mathcal{O}(\log^4 n + C(B) \cdot \log n)$, where $C(B)$ is the cost of encrypting/decrypting a block of $B$ bits}.
\end{lemma}
\begin{proof}
During the \access phase, Algorithm \ref{alg:access}, the client initially probes \historical with \level and \posll queries (Lines~2 and~3). 
By Lemma \ref{lem:histmem_query_app}, both queries require $\mathcal{O}(\log^2 n)$ time.

During the \rebuild phase, Algorithm \ref{alg:rebuild}, we update \historical (Line~4) and use \historical to perform eviction (Line~5).
By Lemma~\ref{lem:histmem_update_app}, the former procedure requires amortized $\mathcal{O}(\log^2 n)$ time.
A \rebuild at level performs $l$ \evict operations.
Each \evict procedure, Algorithm \ref{alg:evict}, at level $l$, executes $2\cdot n_l$ $\level$ queries. 
As a \rebuild at level $l$ occurs every $n_l$ accesses, by Lemma \ref{lem:histmem_update_app}, the amortized time cost incurred by \historical is $\mathcal{O}(l\cdot \log^2(n))$.
Summing over levels $\{1,\ldots, \log n -1\}$, the total cost is
\begin{align*}
    \sum_{l=1}^{L-1} \mathcal{O}(l\cdot \log^2 n) = \mathcal{O}(L^2 \log^2 n) = \mathcal{O}(\log^4 n)\,.
\end{align*}
This runtime dominates the cost incurred by \historical during the \access phase.

{During interactions with the server, the client encrypts and decrypts $O(\log n)$ blocks, amortized, giving the additional cost of $O(C(B)\cdot \log n)$, where $C(B)$ is the time cost of encrypting/decrypting a block of $B$ bits.}  
\end{proof}

{
When compared to a traditional \hierarchical scheme, $\mathcal{O}(\log^4 n)$ represents the additional cost incurred by the use of \historical.
This cost is added to a logarithmic number of encryption/decryption operations per each $B$-bit size block. In practice, the latter is likely to be the dominating cost incurred by the client, while still being lower than the bandwidth time. 
Recall that the encryption/decryption costs are inherent to all ORAM solutions.
}

\subsection{Security}
In Algorithms \ref{alg:access}-\ref{alg:evict}, the red lines indicate operations that involve communication with the server.
To demonstrate the obliviousness of \rankoram, we need to show that its access pattern is independent of the input sequence.

\begin{lemma}
\textup{\rankoram} is oblivious 
\label{lem:security}
\end{lemma}
\begin{proof}
{We follow the standard security argument for \hierarchical constructions since \rankoram adopts the same template (Section~\ref{subsec:hierarchical}).
Recall that the security of accesses relies on the following  invariant: a block is accessed only \textit{once} in a given hash table instance. 
Then as long as each block is located at a location independent of the block content and address, any block retrieved from a level appears as a random index access to the adversary~\cite{goldreich1996software}.
As blocks are mapped to locations using pseudorandom functions, we only need to show that \rankoram satisfies the invariant on one-time-retrieval for real and dummy accesses. 
For real accesses, when a specific address is retrieved it is always inserted into the top of the hierarchy and a rebuild occurs. 
Therefore, when it is next accessed, the block will be located in a new hash table instance.
Similarly, all dummy fetches are ``fresh'' and retrieve an untouched physical address determined by the pseudo random permutation {and a unique counter}.}

To complete the proof, we need to establish the security of the \rebuild method.
A rebuild involves two interactions with the server.
{First, we construct the input array for the shuffle through the $\evict$~routine (Algorithm~\ref{alg:evict}).
This algorithm combines all the elements that were not accessed by the client and hence are known to the adversary already.
The number of these untouched indices is deterministic and is fixed in size ($n_l$ for level $l$). 
The order in which they are accessed is irrelevant and independent of the block content.
Second, we perform a shuffle using an \textit{oblivious} shuffle that is data-independent based on its security definition.
}

{Given that we use pseudorandom functions, it follows that the access patterns for an arbitrary pair of input sequences are \textit{computationally} indistinguishable. }
\end{proof}

Now we have all the components of Theorem~\ref{thm:main}.
Security is given by Lemma~\ref{lem:security}.
The memory cost is incurred by the underlying shuffling algorithm (Lemma~\ref{LEM:SHORTQUEUE}) and the client-side data structure (Lemma~\ref{lem:hist_mem_app}).
Finally, the bandwidth cost is secured by Lemma~\ref{lem:concrete_bandwidth}.

\section{Experiments}
\label{sec:experiments}

In this paper we have presented \rankoram; a Hierarchical ORAM scheme based on a novel client-side data structure, \historical.
\rankoram trades off client memory to achieve bandwidth efficiency.
{We focus the experimental evaluation on comparing the overhead of methods for storing and maintaining metadata at the client.
We have implemented two baseline approaches, \arraymap and \counter, whose  properties are summarized in Table~\ref{tab:client_ds}.
The former is the standard for storing the metadata of ORAM solutions~\cite{chang2016oblivious,dautrich2014burst,ren2014ring,stefanov2011towards}.
The latter has not been used in an ORAM but mentioned in the context of~\partition~\cite{stefanov2011towards} albeit without an instantiation. 
Hence, our implementation is based on our instantiation described in Appendix~\ref{sec:counter}.
This implementation is of independent interest.
Not that, unlike \historical, \counter requires server-side computation.

For each data structure, we measure \textit{peak client memory} and the \textit{update time} as these measurements depend on the access patterns.
The update time measures the total time per \access.
This would include any query costs that are required to execute the \access. 
We refer the reader to Table \ref{tab:practical}
for the bandwidth costs of these schemes (\counter can be used with both \partition and \ringoram) as the bandwidth performance is determined only by the size of the input (otherwise it would leak information).

In order to evaluate the performance of \historical and baseline approaches in practice,
we use real and synthetic workloads.
The real workloads come from two separate commercial cloud storage network traces.
The first trace, provided by Zhang \textit{et al.}  \cite{zhang2020osca}, is collected on Tencent Cloud Block Storage over a 6 day period.
The average block size is $4$KB.
The second trace, provided by Oe \textit{et al.} \cite{oe2018analysis}, 
is collected on the Fujitsu K5 cloud service.
The properties of the traces are summarized in Table~\ref{tab:problem_size}.

In addition, we have generated synthetic workloads based on uniform and Zipfan distributions.
With the Zipfan, or \textit{skewed}, datasets, we vary the size of the problem instance~$n$,
from $2^{21}$ to $2^{29}$
and the skew parameter $\phi \in \{1.1,1.2,1.3,1.4,1.5\}$, where $1.1$ represents low skew and $1.5$ represents high skew.
The synthetic datasets allow us to test the schemes on average and worst-case access scenarios that
ORAM is designed to protect.
We use two block sizes, 64 bytes and 4KB, simulating the size of a cache line and a page size, respectively.
\begin{table}[t]
    \centering
    \caption{Problem and database sizes for datasets on commercial cloud traces.}
    \begin{tabular}{|c||c|c|} \hline
        & $n$
            & database size\\ \hline \hline
    \textsf{Tencet} 
        & 4 370 466 280
            & $17.5$TB \\ \hline
    \textsf{K5cloud}
        & 1 065 643 040
            & $4.3$TB\\ \hline
    \end{tabular}
    \label{tab:problem_size}
\end{table}

\begin{table*}[t]
    \centering
    \caption{Performance of data structures on real cloud traces with block size $B = 4$KB.
    The memory required for the rebuild phase (i.e., the size of \blockbuffer) is
     $540$MB (0.54 GB) for the \textsf{Tencent} dataset and $260$MB (0.26 GB) for the \textsf{K5cloud} dataset. 
    \historical outperforms all competitors in terms of client-memory size and requires
    less memory than \blockbuffer.
    It is also faster than \counter on both test cases.
    Note that the \arraymap is the standard used in implementations for \partition and \ringoram \cite{chang2016oblivious}
    }
    \begin{tabular}{|c|c||c|c|} \hline
    dataset    
        & data structure  
          & client memory (GB)
            & update time ($\mu $ seconds) \\ \hline \hline
    \textsf{Tencent}     
        & \arraymap 
          & 39.33
            & 0.1 \\ 
    \cline{2-4}    & \counter  
          & 0.56
            & 14.4 \\ 
    \cline{2-4}    &  \historical
          & 0.29
            & 10.6 \\ \hline\hline 
    \textsf{K5cloud}
        & \arraymap 
          & 9.60 
            & 0.1 \\ 
    \cline{2-4}    &  \counter
          & 1.30 
            & 8.1  \\ 
    \cline{2-4}    & \historical 
          & 0.13
            & 2.7 \\ \hline   
    \end{tabular}
    \label{tab:cloud_traces}
\end{table*} 

\begin{figure*}
\centering
\begin{subfigure}{.5\textwidth}
  \centering
  \includegraphics[width=.6\linewidth]{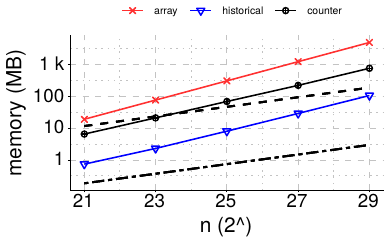}
  \caption{memory vs. $n$ for $\phi=1.2$}
  \label{fig:syn_size}
\end{subfigure}%
\begin{subfigure}{.5\textwidth}
  \centering
  \includegraphics[width=.6\linewidth]{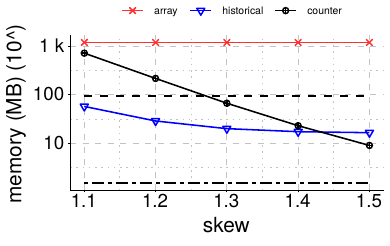}
  \caption{
  memory vs. skew for $n=2^27$
  }
  \label{fig:syn_skew}
\end{subfigure}
\caption{The size of client memory across several ORAM implementations for synthetic workloads of varying size and skew. {Dashed lines correspond to the size of \blockbuffer (temporary client memory needed during reshuffle) for block sizes of 64 bytes and 4 KB. Note that the size of \arraymap, \hierarchical and \counter are independent of the block size.}}
\label{fig:syn_data}
\end{figure*}
Recall that large-client ORAM schemes, including \rankoram, utilize client memory to temporarily store and shuffle $\mathcal{O}(\sqrt{n})$ blocks.
To this end, we also measure the amount of temporary memory required for reshuffles, which we refer to as \blockbuffer, and compare it to the memory requirements of the index data structures stored at the client (i.e., \arraymap, \counter and \historical).
Intuitively, the size of \blockbuffer is the minimum memory requirement of these ORAMs.
Hence, the use of compressed methods, such as \counter and \historical, would be justified only if (1) the memory allocation of \arraymap significantly exceeds the allocation of the \blockbuffer; and (2) the compression methods produce memory allocations less than or equal to the allocations of the \blockbuffer.
Our experiments demonstrate when this is the case.

Prior work has stated that, for a sufficiently large block size, in practice, the memory allocation of the block cache significantly exceeds that of the metadata array.
However, even for $4$KB blocks, experimental work has demonstrated that this is not always the case \cite{chang2016oblivious}.

\subsection{Experimental setup}
All code is written in C\texttt{++}.
We simulate the client and
server on a single machine.
The server is simulated by an interface that abstracts array management:
all data is stored and retrieved on disk (representing the server) and the data structures are stored in RAM (representing the client).

Each workload is executed on a \hierarchical to generate a dataset of accesses and rebuilds.
The client-side data structures are evaluated on these datasets through the metrics of client-memory and update time.
This approach of simulating the accesses and rebuilds at the client allows for a more accurate calculation of update time.

\subsection{Real data}

The results on cloud traces are displayed in Table~\ref{tab:cloud_traces}.
The size of the \blockbuffer, with $B=4$KB, is $540$MB for \textsf{Tencent} and $260$MB for \textsf{K5cloud}.
Notably, the size of \blockbuffer is significantly smaller than the memory allocation for \arraymap, including a factor 80 difference on the \textsf{Tencent} cloud.
This indicates that, particularly for large $n$, the \arraymap represents the significant component of client memory.
\historical outperforms \arraymap and \counter in client memory.
On the \textsf{Tencent} dataset, \historical encodes the block metadata in $0.53$ bits per block and reduces client memory by a factor of $135$ against the baseline \arraymap.
On the \textsf{K5cloud} dataset, \counter encodes the block metadata $9.8$ bits per block and attains a memory allocation that is larger that \historical by a factor of 10. 
The encoding is larger than the $2$ bits per block hypothesized by Stefanov \textit{et al.} \cite{stefanov2011towards} and demonstrates its sensitivity to the access pattern.
Recall that, unlike \historical, \counter has poor worst-case behaviour (see Table \ref{tab:client_ds}).

Both compression techniques, \historical and \counter, obtain memory allocations comparable to the \blockbuffer.
Further, their update times are markedly smaller than a standard network latency of 30-40ms.
Thus, the experiments demonstrate the feasibility of both these techniques in practice.

\subsection{Synthetic Data}
\label{sec:synthticdata}

The results on synthetic data are displayed in Figure \ref{fig:syn_data}.
Both plots contain lines that approximate the size of the \blockbuffer for $B=64$ bytes (long dashed) and $B=4$KB (dashed).
Figure \ref{fig:syn_size} expresses the effect of the problem size for skew parameter $\phi=1.2$.
As expected, the \arraymap is proportional to $\mathcal{O}(n \log n)$ and grows notably faster than the \blockbuffer.
Both \historical and \counter grow linearly with the database size.
For \historical, this is in line with the theoretical bounds.
In contrast, for \counter, this indicates that the worst-case bounds do not hold when there is moderate skew in the access pattern. 
Further, for an Intel SGX secure processor, with an Enclave Page Cache of $96$MB and block size $B=64$ bytes (matching a typical processors cache line),
Figure \ref{fig:syn_size} demonstrates that \rankoram can be executed, on these access patterns, entirely in private memory for $n\le 2^{27}$.
In comparison, \counter fits in private memory for $n<2^{23}$.

To illustrate the effect of skew on the compression techniques, Figure \ref{fig:syn_skew} plots client memory against the skew parameter.
The database size is fixed at $n=2^{27}$.
Both \historical and \counter, on highly skewed access patterns, reduce client memory by a factor of 100 against the baseline \arraymap.
However, the performance of \counter degrades significantly as the amount of skew decreases.
To test this further, we conduct a separate experiment on a uniformly distributed access pattern, which is the worst-case.
In this instance \counter obtained a memory allocation larger than \arraymap ($1.5$ GB for \counter and $1.2$GB for \arraymap when $n=2^{27}$).
In contrast, \historical, still outperformed \arraymap by a factor of $12$ when the access pattern was uniformly distributed.

On all instances of synthetic data, \historical requires a lower memory allocation than the \blockbuffer (the memory required for the rebuild phase).
At the same time, for larger values of $n$, the \arraymap obtains a memory allocation that exceeds the \blockbuffer by a factor of at least $10$.
This testifies to the efficacy of our approach.

One has to be careful when deploying \historical and \counter in practice as varying memory requirements between the skews could introduce an additional side-channel revealing the type of access pattern. To this end, reserving memory for the worst-case is advisable. For such cases, \historical would be preferred due to an order of magnitude smaller memory requirements in the worst case.

\section{Conclusions}
We have presented the first protocol for Hierarchical ORAM that can retrieve
the accessed block in a single-round of communication without requiring server computation.
Our construction, \rankoram, exploits a larger client memory allocation, relative to prior work, to achieve improved bandwidth and latency performance.
The foundation of \rankoram is a novel client-side data structure, \historical, that maintains a compact representation of the locations of the blocks at the server.
Significantly, \historical can be used in \textit{any} Hierarchical ORAM to reduce the number of round trips of communication, per \access, from $\log n$ to one.
{Further, with \historical levels at the server can be stored as permuted arrays, avoiding complex hash tables
and allowing fast and practical oblivious shuffle algorithms to be used for rebuilds.} 

Compared to state-of-the-art passive solutions, \partition and \ringoram, we reduce client memory by a logarithmic factor, while maintaining comparable bandwidth and latency performance.
The standard for passive ORAMs is to use an array to store position maps at the client.
Our experiments, on real network file system traces, demonstrate a reduction in client memory by a factor of a 100 compared to the array approach and by a factor of 10 compared to closest related work.

\bibliographystyle{unsrt}  

\begin{thebibliography}{10}

\bibitem{shinde2016preventing}
Shweta Shinde, Zheng~Leong Chua, Viswesh Narayanan, and Prateek Saxena.
\newblock Preventing page faults from telling your secrets.
\newblock In {\em Proceedings of the 11th ACM on Asia Conference on Computer
  and Communications Security}, pages 317--328, 2016.

\bibitem{wang2017leaky}
Wenhao Wang, Guoxing Chen, Xiaorui Pan, Yinqian Zhang, XiaoFeng Wang, Vincent
  Bindschaedler, Haixu Tang, and Carl~A Gunter.
\newblock Leaky cauldron on the dark land: Understanding memory side-channel
  hazards in sgx.
\newblock In {\em ACM SIGSAC 2017}, pages 2421--2434, 2017.

\bibitem{xu2015controlled}
Yuanzhong Xu, Weidong Cui, and Marcus Peinado.
\newblock Controlled-channel attacks: Deterministic side channels for untrusted
  operating systems.
\newblock In {\em 2015 IEEE Symposium on Security and Privacy}, pages 640--656.
  IEEE, 2015.

\bibitem{islam2012access}
Mohammad~Saiful Islam, Mehmet Kuzu, and Murat Kantarcioglu.
\newblock Access pattern disclosure on searchable encryption: ramification,
  attack and mitigation.
\newblock In {\em Ndss}, volume~20, page~12. Citeseer, 2012.

\bibitem{kellaris2016generic}
Georgios Kellaris, George Kollios, Kobbi Nissim, and Adam O'neill.
\newblock Generic attacks on secure outsourced databases.
\newblock In {\em Proceedings of the 2016 ACM SIGSAC Conference on Computer and
  Communications Security}, pages 1329--1340, 2016.

\bibitem{cash2015leakage}
David Cash, Paul Grubbs, Jason Perry, and Thomas Ristenpart.
\newblock Leakage-abuse attacks against searchable encryption.
\newblock In {\em Proceedings of the 22nd ACM SIGSAC conference on computer and
  communications security}, pages 668--679, 2015.

\bibitem{goldreich1996software}
Oded Goldreich and Rafail Ostrovsky.
\newblock Software protection and simulation on oblivious rams.
\newblock {\em Journal of the ACM (JACM)}, 43(3):431--473, 1996.

\bibitem{larsen2018yes}
Kasper~Green Larsen and Jesper~Buus Nielsen.
\newblock Yes, there is an oblivious {RAM} lower bound!
\newblock In {\em Annual International Cryptology Conference}, pages 523--542.
  Springer, 2018.

\bibitem{asharov2020optorama}
Gilad Asharov, Ilan Komargodski, Wei-Kai Lin, Kartik Nayak, Enoch Peserico, and
  Elaine Shi.
\newblock {OptORAMa}: Optimal oblivious {RAM}.
\newblock {\em EUROCRYPT 2020}, page 403, 2020.

\bibitem{asharov2021oblivious}
Gilad Asharov, Ilan Komargodski, Wei-Kai Lin, and Elaine Shi.
\newblock Oblivious {RAM} with worst-case logarithmic overhead.
\newblock In {\em Annual International Cryptology Conference}, pages 610--640.
  Springer, 2021.

\bibitem{chan2017oblivious}
T-H~Hubert Chan, Yue Guo, Wei-Kai Lin, and Elaine Shi.
\newblock Oblivious hashing revisited, and applications to asymptotically
  efficient {ORAM} and {OPRAM}.
\newblock In {\em International Conference on the Theory and Application of
  Cryptology and Information Security}, pages 660--690. Springer, 2017.

\bibitem{kushilevitz2012security}
Eyal Kushilevitz, Steve Lu, and Rafail Ostrovsky.
\newblock On the (in) security of hash-based oblivious {RAM} and a new
  balancing scheme.
\newblock In {\em Proceedings of the twenty-third annual ACM-SIAM symposium on
  Discrete Algorithms}, pages 143--156. SIAM, 2012.

\bibitem{lu2013distributed}
Steve Lu and Rafail Ostrovsky.
\newblock Distributed oblivious {RAM} for secure two-party computation.
\newblock In {\em Theory of Cryptography Conference}, pages 377--396. Springer,
  2013.

\bibitem{goodrich2011privacy}
Michael~T Goodrich and Michael Mitzenmacher.
\newblock Privacy-preserving access of outsourced data via oblivious {RAM}
  simulation.
\newblock In {\em International Colloquium on Automata, Languages, and
  Programming}, pages 576--587. Springer, 2011.

\bibitem{goodrich2012privacy}
Michael~T Goodrich, Michael Mitzenmacher, Olga Ohrimenko, and Roberto Tamassia.
\newblock Privacy-preserving group data access via stateless oblivious {RAM}
  simulation.
\newblock In {\em Proceedings of the twenty-third annual ACM-SIAM symposium on
  Discrete Algorithms}, pages 157--167. SIAM, 2012.

\bibitem{patel2018panorama}
Sarvar Patel, Giuseppe Persiano, Mariana Raykova, and Kevin Yeo.
\newblock {PanORAMa}: Oblivious {RAM} with logarithmic overhead.
\newblock In {\em 2018 IEEE 59th Annual Symposium on Foundations of Computer
  Science (FOCS)}, pages 871--882. IEEE, 2018.

\bibitem{pinkas2010oblivious}
Benny Pinkas and Tzachy Reinman.
\newblock Oblivious {RAM} revisited.
\newblock In {\em Annual cryptology conference}, pages 502--519. Springer,
  2010.

\bibitem{williams2012single}
Peter Williams and Radu Sion.
\newblock Single round access privacy on outsourced storage.
\newblock In {\em Proceedings of the 2012 ACM conference on Computer and
  communications security}, pages 293--304, 2012.

\bibitem{boneh2011remote}
Dan Boneh, David Mazieres, and Raluca~Ada Popa.
\newblock Remote oblivious storage: Making oblivious {RAM} practical.
\newblock 2011.

\bibitem{fletcher2015bucket}
Christopher~W Fletcher, Muhammad Naveed, Ling Ren, Elaine Shi, and Emil
  Stefanov.
\newblock Bucket {ORAM}: Single online roundtrip, constant bandwidth oblivious
  {RAM}.
\newblock {\em IACR Cryptol. ePrint Arch.}, 2015:1065, 2015.

\bibitem{nagarajan2019rho}
Chandrasekhar Nagarajan, Ali Shafiee, Rajeev Balasubramonian, and Mohit Tiwari.
\newblock $\rho$: Relaxed hierarchical {ORAM}.
\newblock In {\em Proceedings of the Twenty-Fourth International Conference on
  Architectural Support for Programming Languages and Operating Systems}, pages
  659--671, 2019.

\bibitem{hemenway2021alibi}
Brett Hemenway~Falk, Daniel Noble, and Rafail Ostrovsky.
\newblock Alibi: A flaw in cuckoo-hashing based hierarchical {ORAM} schemes and
  a solution.
\newblock In {\em Advances in Cryptology - {EUROCRYPT} 2021 - 40th Annual
  International Conference on the Theory and Applications of Cryptographic
  Techniques}, pages 338--369. Springer, 2021.

\bibitem{holland2020recency}
William~L Holland, Anthony Wirth, and Justin Zobel.
\newblock Recency queries with succinct representation.
\newblock In {\em 31st International Symposium on Algorithms and Computation
  (ISAAC 2020)}. Schloss Dagstuhl-Leibniz-Zentrum f{\"u}r Informatik, 2020.

\bibitem{dautrich2014burst}
Jonathan Dautrich, Emil Stefanov, and Elaine Shi.
\newblock Burst $\{$ORAM$\}$: Minimizing $\{$ORAM$\}$ response times for bursty
  access patterns.
\newblock In {\em 23rd USENIX Security Symposium (USENIX Security 14)}, pages
  749--764, 2014.

\bibitem{stefanov2013path}
Emil Stefanov, Marten Van~Dijk, Elaine Shi, Christopher Fletcher, Ling Ren,
  Xiangyao Yu, and Srinivas Devadas.
\newblock Path {ORAM}: an extremely simple oblivious {RAM} protocol.
\newblock In {\em Proceedings of the 2013 ACM SIGSAC conference on Computer \&
  communications security}, pages 299--310, 2013.

\bibitem{stefanov2011towards}
Emil Stefanov, Elaine Shi, and Dawn Song.
\newblock Towards practical oblivious {RAM}.
\newblock {\em NDSS}, 2012.

\bibitem{ren2014ring}
Ling Ren, Christopher~W Fletcher, Albert Kwon, Emil Stefanov, Elaine Shi,
  Marten van Dijk, and Srinivas Devadas.
\newblock Ring {{ORAM}}: Closing the gap between small and large client storage
  oblivious {RAM}.
\newblock {\em IACR Cryptol. ePrint Arch.}, 2014:997, 2014.

\bibitem{holland2022waksman}
William~L Holland, Olga Ohrimenko, and Anthony Wirth.
\newblock Efficient oblivious permutation on the waksman network.
\newblock In {\em AsiaCCS 2022}, 2022.

\bibitem{garg2016tworam}
Sanjam Garg, Payman Mohassel, and Charalampos Papamanthou.
\newblock {TWORAM}: efficient oblivious {RAM} in two rounds with applications
  to searchable encryption.
\newblock In {\em Annual International Cryptology Conference}, pages 563--592.
  Springer, 2016.

\bibitem{shi2011oblivious}
Elaine Shi, T-H~Hubert Chan, Emil Stefanov, and Mingfei Li.
\newblock Oblivious ram with o ((logn) 3) worst-case cost.
\newblock In {\em Asiacrypt}, volume 7073, pages 197--214. Springer, 2011.

\bibitem{cao2021streamline}
Dingyuan Cao, Mingzhe Zhang, Hang Lu, Xiaochun Ye, Dongrui Fan, Yuezhi Che, and
  Rujia Wang.
\newblock Streamline ring oram accesses through spatial and temporal
  optimization.
\newblock In {\em 2021 IEEE International Symposium on High-Performance
  Computer Architecture (HPCA)}, pages 14--25. IEEE, 2021.

\bibitem{raoufi2023ab}
Mehrnoosh Raoufi, Jun Yang, Xulong Tang, and Youtao Zhang.
\newblock Ab-oram: Constructing adjustable buckets for space reduction in ring
  oram.
\newblock In {\em 2023 IEEE International Symposium on High-Performance
  Computer Architecture (HPCA)}, pages 361--373. IEEE, 2023.

\bibitem{chang2016oblivious}
Zhao Chang, Dong Xie, and Feifei Li.
\newblock Oblivious {RAM}: A dissection and experimental evaluation.
\newblock {\em Proceedings of the VLDB Endowment}, 9(12):1113--1124, 2016.

\bibitem{raman2007succinct}
Rajeev Raman, Venkatesh Raman, and Srinivasa~Rao Satti.
\newblock Succinct indexable dictionaries with applications to encoding k-ary
  trees, prefix sums and multisets.
\newblock {\em ACM Transactions on Algorithms (TALG)}, 3(4):43--es, 2007.

\bibitem{oprea2007integrity}
Alina Oprea and Michael~K Reiter.
\newblock Integrity checking in cryptographic file systems with constant
  trusted storage.
\newblock In {\em USENIX Security Symposium}, pages 183--198. Boston, MA;,
  2007.

\bibitem{bradley1969optimizing}
Stevan~D Bradley.
\newblock Optimizing a scheme for run length encoding.
\newblock {\em Proceedings of the IEEE}, 57(1):108--109, 1969.

\bibitem{stefanov2012fastprp}
Emil Stefanov and Elaine Shi.
\newblock Fastprp: Fast pseudo-random permutations for small domains.
\newblock {\em Cryptology ePrint Archive}, 2012.

\bibitem{ohrimenko2014melbourne}
Olga Ohrimenko, Michael~T Goodrich, Roberto Tamassia, and Eli Upfal.
\newblock The melbourne shuffle: Improving oblivious storage in the cloud.
\newblock In {\em International Colloquium on Automata, Languages, and
  Programming}, pages 556--567. Springer, 2014.

\bibitem{patel2017cacheshuffle}
Sarvar Patel, Giuseppe Persiano, and Kevin Yeo.
\newblock Cacheshuffle: An oblivious shuffle algorithm using caches.
\newblock {\em arXiv preprint}, 2017.

\bibitem{elias1975universal}
Peter Elias.
\newblock Universal codeword sets and representations of the integers.
\newblock {\em IEEE transactions on information theory}, 21(2):194--203, 1975.

\bibitem{zhang2020osca}
Yu~Zhang, Ping Huang, Ke~Zhou, Hua Wang, Jianying Hu, Yongguang Ji, and Bin
  Cheng.
\newblock $\{$OSCA$\}$: An $\{$Online-Model$\}$ based cache allocation scheme
  in cloud block storage systems.
\newblock In {\em 2020 USENIX Annual Technical Conference (USENIX ATC 20)},
  pages 785--798, 2020.

\bibitem{oe2018analysis}
K~Oe, K~Ogihara, and T~Honda.
\newblock Analysis of commercial cloud workload and study on how to apply cache
  methods.
\newblock {\em IEICE CPSY, Kumamoto, Japan (SWoPP 2018), CPSY2018-14}, pages
  7--12, 2018.

\bibitem{foster1973generalization}
Caxton~C Foster.
\newblock A generalization of avl trees.
\newblock {\em Communications of the ACM}, 16(8):513--517, 1973.

\bibitem{navarro2014wavelet}
Gonzalo Navarro.
\newblock Wavelet trees for all.
\newblock {\em Journal of Discrete Algorithms}, 25:2--20, 2014.

\end{thebibliography}



\appendix

\section{Compressing Position Maps}
\label{sec:counter}

Stefanov \etal, with \partition, were the first to suggest compressing client-side metadata to support a client-memory efficient ORAM protocol \cite{stefanov2011towards}. 
Their method constitutes a key baseline for \historical.
However, the authors omit details for implementing the approach.
To supplement their work, {and  to provide a suitable comparison for empirical evaluation},
we provide those details here.
We begin with an overview of their concept.

\partition partitions the database into $\sqrt{n}$ smaller Hierarchical ORAMs of size $\sim \sqrt{n}$.
Each block is randomly assigned to a point in the partition. 
The position map of \partition stores the following pieces of metadata for a block: (1) the partition number; (2) the level number; and (3) the hash table offset.
The partition numbers are selected uniformly at random and thereby have high entropy.
Consequently,
to achieve a more compressible position map,
the count of each block is stored, instead of the partition number, and used as input to a pseudorandom function.
For example, let $\textsf{ctr}_i$ be the count of block $i$ and $\textsf{PRF}$ denote the function.
Then block $i$ is assigned to partition $\textsf{PRF}(i\mid \textsf{ctr}_i)$.
On each access, the count is incremented and the psuedorandom function generates a fresh, and seemingly random, partition number.
Unlike recency, the count is fixed between accesses.
Therefore, \historical cannot be used by \partition to generate partition numbers.
The advantage of storing the counts is that they are highly compressible for sequential access patterns.
This is demonstrated by Opera \textit{et al.} \cite{oprea2007integrity}, who outline a compression method designed to leverage sequentiallity.

Further, Stefanov \textit{et al.} note that, if all levels are full, each block has probability $2^{l-L}$ of being in level $l$.
Thus, the level information has low entropy and is highly compressible. 
No compression algorithm is nominated.
Lastly, the hash table location metadata is dispensed with by uploading the blocks, including dummies, with random ``aliases''.
Then, during retrieval, the client requests blocks by their alias and the server finds the block on the clients behalf.
Consequently, server-side computation is introduced.
We refer to this combined approach of compressing metadata as \counter.
It requires two data structures; one for compressing the block counters; and one for compressing the level information.
We now provide an instantiation of both.

\subsection{Data structures}

A method for compressing counters involves storing \textit{counter intervals} instead of a separate count for each block \cite{oprea2007integrity}.
A counter interval stores a single count for each interval of consecutive blocks with the same \access count.
The data structure contains two arrays.
An index array \textsf{ind} stores the starting index for each interval and a counter array \textsf{ctr} stores the count for each interval.
Thus the count for an index $i\in [\textsf{ind}[j], \textsf{ind}[j+1])$ has count $\textsf{ctr}[j]$.

The challenge is to keep the arrays \textit{compact} under a dynamic workload.
We want to avoid resizing the array at each update.
Thus, we divide each array into segments of width $[\tfrac{1}{2} Z, Z]$ for some parameter $Z$.
Each segment is implemented with a \textit{dynamic} array that we resize in accordance with its current capacity. 
Thus, to keep the memory allocation tight, at most one segment is resized after each \access.
The segments are stored in the leaves of a balanced binary search tree (we use an AVL tree \cite{foster1973generalization}).
When the combined cardinality of adjacent segments falls below $Z$, we merge the segments.
Inversely, when the cardinality of a segment exceeds $Z$, we split the segment.
The mechanics of the tree implementation keep the structure balanced.
The parameter $Z$ invokes a trade-off;
small $Z$ results in fast updates, as the smaller segment requires less shifting of elements and reallocation of memory, 
but incurs a large tree structure.
On the other hand, large $Z$ induces slow updates but a small auxiliary tree structure.

The level information can be stored in any dynamic string implementation, such as a wavelet tree \cite{navarro2014wavelet}.
We adopt a run-length encoding, similar to the constituent dictionaries in Section \ref{SEC:CLIENT}.
To maintain a dynamic run-length code, we apply the same method for maintaining counter intervals.
The code is split into segments, implemented with dynamic arrays, and stored in the leaves of a balanced binary search tree.

\subsection{Performance}

For completeness, we provide a theoretical bound for \counter.
The method is suited for sequential access patterns and has poor worst-case behaviour. 
This could limit its application in memory constrained environments under dynamic workloads with changing distributions.
In the worst-case, if the memory allocation exceeds the application bounds and this scenario is observable to the adversary, it introduces an additional side-channel. 

\begin{lemma}
For a database of size $n$, for $Z=\Theta(\log n)$, to store the block frequencies on a workload of length $n\cdot\textup{\textsf{poly}}(n)$, \textup{\counter} requires $\mathcal{O}(n \log n) $ bits.
\label{lem:counter}
\end{lemma}

\begin{proof}
    We construct a pathological workload that induces the theoretical bound.
Consider an access pattern $A=\langle1,3,5,\ldots, n-1\rangle$ that is executed $c=\textsf{poly}(n)$ number of times.
$A$ contains $n$ distinct counter intervals:
$[(0,0),(1,c),(2,0),(3,c),\ldots]$.
As the values of the frequencies are polynomial in size, each interval requires $\mathcal{O}(\log n)$ bits to store.
Thus, the total memory allocation for the counter intervals is $\mathcal{O}(n\log n)$ bits.

As each segment contains $\Theta(\log n)$ intervals, the auxiliary binary search tree contains $\mathcal{O}(n/\log n)$ nodes and, with $\mathcal{O}(\log n)$ bit pointers, occupies $\mathcal{O}(n)$ bits.

The memory allocation of the compressed level information can be computed from its entropy.
Let $X_i$ be a random variable that denotes the current level of block $i$.
We assume that the access pattern is uniformly distributed, as this is the worst-case for run-length codes.
If all levels are filled, then 
\[
\Pr[X_i=l] = 2^{l-L}\,.
\]
Thus, the entropy of $X_i$ is
\begin{align}
   H(X_i) = -\sum_{l=0}^{L-1} \log_2(2^{l-L}) = \sum_{i=1}^L i\cdot 2^{-i} < 1\,.
   \label{eqn:entropy}
\end{align}
As we calculate the run lengths $\mathbf{R}=(R_1,R_2, \ldots)$ from the block levels $\mathbf{X}=(X_1,\ldots,X_n)$, it holds that
\[
H(\mathbf{R}) \leq H(\mathbf{X}) \leq \sum_{i=1}^n H(X_i) = n\cdot H(X_i) < n\,,
\]
by Equation \eqref{eqn:entropy}.
Thus, the run length encoding can be stored in less than $n$ bits.
Similar to the counter intervals, the auxiliary tree, with $Z=\Theta(\log n)$, has $\mathcal{O}(n/\log n)$ nodes and occupies $\mathcal{O}(n)$ bits.
\end{proof}
\subsection{Parameter Tuning}

Performance of \counter depends on a parameter $Z$ that we tune as follows and use for the main experiments in Section~\ref{sec:experiments}.
\begin{figure}[t]
    \centering
    \includegraphics[width=0.3\linewidth]{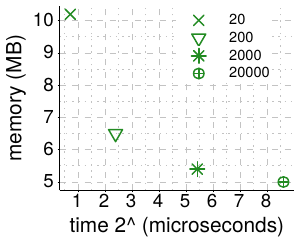}
    \caption{Client memory vs. update time for \counter on $Z\in\{20,200,2000,20000\}$.
    The parameter $Z$ represents the size of the dynamic array segments in the component data structures of \counter.}
    \label{fig:counter_Z}
\end{figure}

Recall that the core of the \counter data structure is a dynamic array, for the counters, and a dynamic string of run lengths, for the levels.
Both dynamic structures are split into segments of size $[\tfrac{1}{2}Z, Z]$, for a parameter $Z$, and the segments are stored in the leaves of a balanced binary search tree.
To infer the effect of $Z$ on performance we tested \counter on a synthetic dataset ($n=2^{21}$ and $\phi$ = $1.2$ as described in Section~\ref{sec:synthticdata})
for values $Z\in\{20,200,2000,20000\}$.
The results are displayed in Figure \ref{fig:counter_Z}.
The test demonstrates a clear trade-off between client memory and update time.
For $Z = 20$ updates are fast, as only a small segment is updated on each \access.
However, due to the size of the auxiliary binary tree, the memory allocation is close to double the other instances of \counter.
As $Z$ increases, the memory drops sharply and begins to plateau for $Z>20$.
In contrast, the update time steadily increases.
To achieve a good balance between memory and throughput, we select $Z=20$ for all experiments in Section~\ref{sec:experiments}.

\section{Shuffling in \rankoram}
\label{sec:shuffling}

\begin{algorithm}[t]
    \SetAlgoLined
    \DontPrintSemicolon
    \SetKwInOut{internal}{Internal state}{}{}
    \SetKwProg{myproc}{define}{}{}
        \myproc{\textup{\shortqueue{($\pi, \mathrm{I}$)}}}
    {
        $\phi \gets$ random permutation to assign array indicies to buckets\;
        Initialize temporary arrays $\mathrm{T}_1,\ldots,\mathrm{T}_{\sqrt{n}}$ at the server\;
        Initialize the queues $Q_1,\ldots, Q_{\sqrt{n}}$ at the client\;
        \For{$i \in \{1,\ldots,\sqrt{n}\}$}
        {
            $R \gets \{\phi(j) \mid j \in \{(i-1)\cdot \sqrt{n}, \ldots, i\cdot\sqrt{n}-1\}\}$ \;
            $I_i \gets \mathrm{I}[R]$ \;
            \For{$x \in I_i$}
            {
            $d\gets \lfloor \pi(x)/\sqrt{n}\rfloor$ \;
            $Q_d.\textsf{push}(x)$\;
            }
            \For{$j \in \{1,\ldots,\sqrt{n}\}$}
            {
                \Repeat{\textup{Twice}}
                {
                \If{$Q_j$ \textup{is non-empty}}
                {
                    Write element from $Q_j$ to $\mathrm{T}_j$\;
                }
                \Else 
                {
                    Write dummy element to $\mathrm{T}_j$ \;
                }
                }
            }
        }
        Initialize output array $\mathrm{O}$ from the server \;
            \For{$i \in \{1,\ldots,\sqrt{n}\}$}
            {
                $T \gets $ download $T_i$ from the server \;
                $T \gets T \cup Q_i$\;
                Remove dummies from $T$, order elements according to $\pi$ and insert ``fresh'' dummies.\;
                $\mathrm{O}[\{(i-1)\cdot \sqrt{n}, \ldots, i\cdot2\sqrt{n}-1\}] \gets T$\;
            }
        \KwRet $\mathrm{O}$\;
    }
\caption{\shortqueue oblivious dummy shuffle method. 
The algorithm takes as input a permutation $\pi:[2n]\rightarrow[2n]$ and an array $\mathrm{I}$ of length $n$ stored remotely at the server.
}
\label{alg:shuffle}
\end{algorithm}

For a \rebuild into level $l$, our algorithm, named \shortqueue, takes as input an array of $n_l$ untouched blocks (possibly including dummies) and produces an output array of length $2n_l$.
The output array contains a random shuffling of the untouched blocks plus an additional $n_l$ dummy blocks.
This is a variation of the functionality of Definition \ref{def:rand_perm}.

For the sake of generalization, for the remainder of the exposition, we set $n:= n_l$ and $\pi := \pi_l(S_l.\rank)$. 
Similar to \cacheshuff \cite{patel2017cacheshuffle}, \shortqueue uniformly at random assigns the indicies of the input array into $\sqrt{n}$ buckets each of equal size $C$ ($\sim \sqrt{n}$).
Let $I_1, \ldots, I_{\sqrt{n}}$ denote the buckets of the input indices.
The client also initializes a temporary array at the server, divided into $\sqrt{n}$ chucks of size $2C$.
Let $T_1, \ldots, T_{\sqrt{n}}$ denote the chunks of the temporary array.
The client initializes, in private memory, the queues $Q_1,\ldots,Q_{\sqrt{n}}$.
Through $\sqrt{n}$ rounds, the client performs the following operations.
For round $j$:
\begin{enumerate}
    \item Download the input chunk $I_j$ into private memory.
    \item For each \textit{real} block $x \in I_j$, let $d = \lfloor\pi(x)/{\sqrt{n}}\rfloor$, and place $x$ in queue $Q_d$.
    \item In $\sqrt{n}$ rounds, for each queue $Q_k$, place two blocks in $T_k$. 
    If the queue is empty, place dummy blocks.
\end{enumerate}
At the conclusion of this subroutine, all untouched blocks are either in the correct chunk in the temporary array, that is, a chunk that contains its final destination index, or the correct queue.
The routine is named ``\textit{short} queue shuffle'' as the arrival rate for each queue is half the departure rate.
Subsequently, the client, in consecutive rounds, downloads each chunk in the temporary array; combines the chunk with any remaining blocks in its corresponding queue; arranges the real blocks according to $\pi$; fills empty spaces with fresh dummies; and uploads the shuffled chunk to the output array. 
Psuedocode is provided in Algorithm~\ref{alg:shuffle}.

The procedure is oblivious as the access pattern at the initial downloading of the input buckets does not depend on the input and the remaining accesses are identical for all inputs of the same length.
However, if the combined size of the queues becomes $\omega(\sqrt{n})$, we exceed our memory threshold and the algorithm fails.
Fortunately, this happens only with negligible probability.
We summarize performance with the following Lemma.

\lemshortq*

Due to its similarities to Lemma 4.2 in \cite{patel2017cacheshuffle} and space constraints, the proof is placed in the full version of the work.
Informally, the algorithm requires  $\sqrt{n}$ round-trips as each step can be into an upload/download of  $\sqrt{n}$ blocks.
For example, the step where blocks are iteratively placed in the temporary array segments can be completed in  a single batch.
For the bandwidth, the first component of the algorithm downloads $n$ addresses and uploads $2n$ addresses.
The second component, in a sequence of rounds, downloads the full temporary array of length $2n$ and uploads it to the output array.
This leads to a total bandwidth cost of $7n$.

\section{Algorithm \ref{alg:rle_query}}
\label{sec:alg_5}
\begin{algorithm}[]
    \SetAlgoLined
    \DontPrintSemicolon
    \SetKwInOut{internal}{Internal state}{}{}
    \SetKwProg{myproc}{define}{}{}
        \myproc{\textup{$\membership(a, D)$}}
    {
        $i \gets $ return index of binary search on $A$ with input $a$ 
        \tcp*{$a\in [A[i],A[i+1])$}
        $\textsf{ptr} \gets P[i]$ \;
        $\textsf{temp} \gets $ return address at location \textsf{ptr} in $D$ \;
        \While{$a \leq \textup{\textsf{temp}}$}
        {
            \If{$a = \textup{\textsf{temp}}$}
            {
                \KwRet \textsf{true}\;
            }
            \Else 
            {
                \tcp*{move to the next element in $D$}
                $\textsf{ptr} \gets \textsf{ptr}+1$\; 
            $\textsf{temp} \gets $ return address at location \textsf{ptr} in $D$\;
            }
        }
        \KwRet \textsf{false}\;
    }
     \myproc{\textup{$\indll(i, D)$}}
    {
        $j \gets \lfloor i/W \rfloor$ \;
        $i^{\prime} \gets i - j\cdot W$\;
        $\textsf{ptr} \gets P[j]$ \;
        $\textsf{ptr} \gets \textsf{ptr}+i^{\prime}$ \tcp*{iterate through $i^{\prime}$ prefix codes}
        \KwRet address at location \textsf{ptr} in $D$\;
    }
\caption{Query algorithms for a run-length code dictionary $D$. 
The code is split into segments of size $W$ and supported by the auxiliary arrays $A$, where $A[i]$ denotes the virtual address at the beginning of segment $i$, and $P$, where $P[i]$ denotes a pointer to segment $i$.}
\label{alg:rle_query}
\end{algorithm}

\end{document}